%% file: DK18aczysty.tex
\documentclass[envcountsame]{llncs}
\usepackage[T1]{fontenc}
\usepackage[frak=pxtx,cal=pxtx]{mathalfa}

\makeatletter
\let\c@remark\relax
\let\c@lemma\relax
\let\c@definition\relax
\let\c@corollary\relax

\usepackage{chngcntr}

\spnewtheorem{theorem}{Theorem}{\bfseries}{\itshape}
\counterwithin{theorem}{section}
\spnewtheorem{lemma}[theorem]{Lemma}{\bfseries}{\itshape}
\spnewtheorem{definition}[theorem]{Definition}{\bfseries}{\rmfamily}
\spnewtheorem{claim}[theorem]{Claim}{\bfseries}{\itshape}
\spnewtheorem{corollary}[theorem]{Corollary}{\bfseries}{\rmfamily}

\usepackage{enumerate}
\usepackage{tikz}
\usetikzlibrary{calc}
\usetikzlibrary{decorations.pathreplacing}
\usetikzlibrary{decorations.pathmorphing}
\usetikzlibrary{decorations.text}

\usepackage{amsmath,amsfonts,url,amssymb, latexsym}
\usepackage[english]{babel}
\usepackage{geometry}
\usepackage{pst-node}
\usepackage{hyperref}

\input{macrosEQ.tex}

\newcommand{\keywords}[1]{\par\addvspace\baselineskip
\noindent\keywordname\enspace\ignorespaces#1}

\begin{document}

\title{Unary negation fragment with equivalence relations has the finite model property\thanks{This paper is an extended and improved version of LICS'18 paper \cite{DK18}. In particular, it corrects a minor bug from the conference version, slightly strengthening  the inductive assumption
in Lemma \ref{l:finiteeq}. }}

\author{Daniel Danielski \and Emanuel Kiero\'nski}
\institute{University of Wroc\l aw}

\maketitle

\begin{abstract}
We consider an extension of the unary negation fragment of first-order logic in which arbitrarily many binary symbols
may be required to be interpreted as equivalence relations. We show that this extension has the finite model property.
More specifically, we show  that every satisfiable formula has a model of at most doubly exponential size. We argue that
the satisfiability (= finite satisfiability) problem for this logic is \TwoExpTime-complete. 
We also transfer our results
to a restricted variant of the guarded negation fragment with equivalence relations.
\end{abstract}


\section{Introduction}
A simple yet beautiful idea of restricting negation to subformulas with at most one free variable
led ten Cate and Segoufin to a definition of an appealing fragment of first-order logic, called the unary negation
fragment, \UNFO{} \cite{StC13}. \UNFO{} turns out to have very nice algorithmic and model-theoretic properties,
and, moreover, it has strong motivations from various areas of computer science. 
\UNFO{} has the finite model property: every satisfiable formula has a finite model. This immediately implies the decidability of the satisfiability problem (does a given formula have a model?) and the finite satisfiability problem (does a given formula have a finite model?). To get tight complexity bounds
one can, e.g., use another convenient property of \UNFO, that every satisfiable formula has a tree-like model, and show that satisfiability is \TwoExpTime-complete.
What is interesting, the lower bound holds even for bounded variable versions of this logic, and already the fragment with three variables
is \TwoExpTime-hard. As several other seminal fragments of first-order logic, like the two variable fragment, \FOt{} \cite{Mor75}, the guarded fragment, \GF{} \cite{ABN98}, and
the fluted fragment, \FF{} \cite{P-HST16}, \UNFO{} embeds propositional (multi)-modal logic, which opens connections to, e.g., such fields as verification of hardware and software or knowledge representation. Moreover, in contrast to the  fragments mentioned above, \UNFO{} can express unions of conjunctive queries, which makes it potentially attractive for the database community. 

Similarly to most important decidable fragments of first order logic, including \FOt{}, \GF{}  and \FF{}, \UNFO{} has a drawback,
which seriously limits its potential applications, namely, it cannot express transitivity of a binary relation, nor a related 
property of being an equivalence. This justifies studying  formalisms, equipping the basic logics  with some facilities allowing to express the above mentioned properties. The simplest way to obtain such  formalisms is to divide the signature into two parts, a base part and a distinguished part,
the latter containing only binary symbols, and impose explicitly some semantic constraints on the interpretations of the symbols from the
distinguished part, e.g., require them to be interpreted as equivalences.  Generally, the results are negative: both \FOt{} and \GF{}
become undecidable with equivalences or with arbitrary transitive relations.  More specifically, the satisfiability and the finite satisfiability
problems for \FOt{} and even for the two-variable restriction of \GF, \GFt{}, with two transitive relations \cite{Kie05,Kaz06} or three equivalences \cite{KO12} are undecidable. Also the fluted fragment is undecidable when extended by
equivalence relations [I.~Pratt-Hartmann, W.~Szwast, L.~Tendera, \emph{private communication}].
Positive results were obtained for \FOt{} and \GF{} only when the distinguished signature contains just one transitive symbol \cite{P-H18} or two equivalences \cite{KMP-HT14},
or when some further syntactic restrictions on the usage of distinguished symbols are imposed \cite{ST04,KT18}. 

\UNFO{} turns out to be an exception here, since its satisfiability problem remains decidable in the presence of arbitrarily many 
equivalence or transitive relations. This can be shown by reducing the satisfiability problem for \UNFO{} with equivalences to
\UNFO{} with arbitrary transitive relations (see Lemma \ref{l:red}). The decidability and \TwoExpTime-completeness of the satisfiability problem 
for the latter follow from two independent recent works, respectively by Jung et al.~\cite{JLM18} and by Amarilli et al.~\cite{ABBB16}. 
In the first of them the decidability of \UNFO{} with transitivity is stated explicitly, as a corollary from the decidability of the
unary negation fragment with regular path expressions. The second shows decidability of 
the guarded negation fragment, \GNFO{},  with transitive relations restricted to non-guard positions (for more about this logic see Section 5), which embeds  \UNFO{} with transitive relations.

Both the above mentioned decidability results are obtained by employing  tree-like model properties of the logics and then using some automata techniques. Since tree-like unravelings of models are infinite, such approach works only for general satisfiability, and gives no insight into the decidability/complexity of the finite satisfiability problem.  

In computer science, the importance of decision procedures for finite satisfiability arises from the fact that most objects about which we may want to reason using 
logic are finite. For example, models of programs have finite numbers of states and possible actions and real world databases contain finite sets of facts.
Under such scenarios, an ability of solving only the general satisfiability problem may not be fully satisfactory. 

In this paper we show that \UNFO{} with arbitrarily many equivalence relations, \UNFOEQ, has the {finite model property}. It follows that the finite satisfiability and the general satisfiability problems for the considered logic coincide,
and, due to the above mentioned reduction to \UNFO{} with transitive relations, can be solved in \TwoExpTime. The corresponding lower bound can be
obtained even for the two-variable version of the logic, in the presence of just two equivalence relations.  
We further transfer our results to the intersection of  \GNFO{} with equivalence relations on non-guard positions and the one-dimensional fragment \cite{HK14}.
A formula is \emph{one-dimensional} if its every maximal block of quantifiers leaves at most one variable free.
Moving from \UNFO{} to this restricted variant of \GNFO{}
significantly increases the expressive power.

Studying equivalence relations may be seen as a step towards understanding finite satisfiability of \UNFO{} or \GNFO{} with arbitrary transitive relations.  However, equivalence relations are also interesting on its own and in computer science were studied in various contexts. They play an important role in modal and
epistemic logics, and were considered in the area of interval temporal logics \cite{MPS16}. Data words
 \cite{BDM11} and data trees \cite{BMS09}, studied in the context of XML reasoning use an equivalence relation to compare 
data values, which may come from a potentially infinite alphabet; we remark, that, again, decidability results over such data structures are obtained only in the presence of a single equivalence relation, that is they allow to compare objects only with respect to a single parameter.

\medskip\noindent
{\bf Related work.} There are not too many decidable fragments of first-order logic whose finite satisfiability is known to remain decidable when extended by an unbounded number of equivalence relations. One exception is the two-variable guarded fragment with equivalence guards, \GFtEG, a logic without the finite model property, whose finite satisfiability
is \NExpTime-complete \cite{KT18}. 
\GFtEG{} slightly differs in spirit from the mentioned decidable variant of \GNFO{} with equivalence
relations on non-guard positions, and thus also from \UNFOEQ{} which is a fragment of the latter. We remark however that these two approaches are not completely orthogonal. E.g., 
a \GFtEG{} formula $\forall xy (E(x,y) \rightarrow (P(x) \wedge P(y)))$, in which atom $E(x,y)$ is used as a guard, when treated as a \GNFO{} formula has $E(x,y)$ on a non-guard position; actually, it is a \UNFOEQ{} formula. Simply, guards play slightly different roles in \GF{} and \GNFO{}.

The decidability of the satisfiability problem for both \GFtEG{} and \UNFOEQ{} can be shown relatively easily, by exploiting tree-based model properties for
both logics. The analysis of the corresponding finite satisfiability problems is much more challenging. It turns out that the difficulties arising when considering  \GFtEG{} and \UNFOEQ{} are of different nature. The main problem in the case of \GFtEG{} is that it allows, using guarded occurrences of inequalities $x \not= y$, to restrict some types of elements to appear at most once in every abstraction class of the guarding equivalence relation. 
This causes that some care is needed when
performing surgery on models, and seems to require a global view at some of their properties. Indeed, the solution employs integer programming to describe some 
global constraints on models of the given formula. What is however  worth remarking, in the case of \GFtEG{} one can always construct models in
which every pair of elements is connected by at most one equivalence. So, \GFtEG{} does not allow for a real interaction among equivalence relations.

Inequalities $x \not=y$ are not allowed in \UNFOEQ, and indeed we do not have here any problems with duplicating elements of any type. 
On the other hand, \UNFOEQ{} allows for  a non-trivial interaction among equivalences, and this seems to be the source of main obstacles for finite model constructions.
Surprisingly, such  obstacles are present already in the two-variable version of our logic.
More intuitions about problems arising will be given later. 

Our solution employs a novel (up to our knowledge)
inductive approach to build a finite model of a satisfiable formula, starting from an arbitrary model. In the base of
induction we construct some initial fragments in which none of the equivalences plays an important role. Such fragments are then
joined into bigger and bigger structures, in which more and more equivalences become significant. This process eventually yields a finite model of the given formula.

\medskip\noindent
{\bf Organization of the paper.} Section 2 contains formal definitions and presents some basic facts. 
In Section 3 we show the finite model property for a restricted, two-variable variant of our logic, \UNFOtEQ.
We believe that treating this simpler setting first will help the reader to understand our ideas and techniques,
since it allows them to be presented without some quite complicated technical details appearing in
the general case.  Then in Section 4 we describe
the generalization of our construction working for full \UNFOEQ, pinpointing the main differences and additional difficulties arising in comparison
to the two-variable case.  In Section 5 we transfer our 
results to the one-dimensional guarded negation fragment with equivalences. Section 6 concludes the paper.

\section{Preliminaries}\label{s:preliminaries}

\subsection{Logics and structures}

We employ standard terminology and notation from model theory.
 In particular, we refer to structures using Gothic capital letters,
and their domains using the corresponding Roman capitals.
For a structure $\str{A}$ and $B \subseteq A$ we use 
$\str{A} \restr B$ or $\str{B}$ to denote the restriction of $\str{A}$ to $B$.

We work with purely relational signatures $\sigma=\sigma_{\cBase} \cup \sigma_{\cDist}$ where $\sigma_{\cBase}$ is the \emph{base signature}
and $\sigma_{\cDist}$ is the \emph{distinguished signature}. All symbols from $\sigma_{\cDist}$ are binary.
Over such signatures we define the \emph{unary negation fragment of first-order logic}, \UNFO{} as in \cite{StC13} by the following
grammar:
$$\phi=R(\bar{x}) \mid x=y \mid \phi \wedge \phi \mid \phi \vee \phi \mid \exists x \phi \mid \neg \phi(x) $$
where $R$ represents a relation symbol and, in the last clause, $\phi$ has no free variables besides (at most) $x$. 

A typical formula
not expressible in \UNFO{} is $x \not= y$.
We formally do not allow universal quantification. However we will allow
ourselves to use $\forall \bar{x} \neg \phi$ as an abbreviation for $\neg \exists \bar{x}  \phi$,
for an \UNFO{} formula $\phi$. Note that  $\forall xy \neg P(x,y)$ is in \UNFO{} but  $\forall xy  P(x,y)$ is not.

The \emph{unary negation fragment with equivalences}, \UNFOEQ{} is defined by the same grammar as \UNFO{}. When satisfiability of its
formulas is considered, we restrict the
class of admissible models to those that interpret all symbols from $\sigma_{\cDist}$ as equivalence relations.
We also mention an analogous logic \UNFOTR{} in which the symbols from $\sigma_{\cDist}$ are interpreted
as (arbitrary) transitive relations.

\subsection{Atomic types}

An \emph{atomic $k$-type} (or, shortly, a $k$-\emph{type}) over a signature $\sigma$ is a maximal satisfiable set of literals (atoms and negated atoms) over $\sigma$
with variables $x_1, \ldots, x_k$. 
We will sometimes identify a $k$-type with the conjunction of its elements. 
Given a $\sigma$-structure $\str{A}$ and a tuple $a_1, \ldots, a_k \in A$ we denote by $\type{\str{A}}{a_1, \ldots, a_k}$ the atomic
$k$-type \emph{realized} by $a_1, \ldots, a_k$, that is the unique $k$-type $\alpha(x_1, \ldots, x_k)$ such that $\str{A} \models \alpha(a_1, \ldots, a_k)$.

\subsection{Normal form and witness structures}

We say that an \UNFOEQ{} formula is in Scott-normal form if it is of the shape
\begin{eqnarray}
\forall x_1, \ldots, x_t \neg \phi_0(\bar{x}) \wedge \bigwedge_{i=1}^{m} \forall x \exists \bar{y} \phi_i(x,\bar{y})
\label{eq:nf}
\end{eqnarray}
where each $\phi_i$ is an \UNFOEQ{} quantifier-free formula.
This kind of normal form was introduced in the bachelor's thesis \cite{Dzi17}. 

\begin{lemma} \label{l:nf}
For any \UNFOEQ{} formula $\phi$ one can compute in polynomial time a normal form \UNFOEQ{} formula
$\phi'$ over signature extended by some fresh unary symbols, such that 
any model of $\phi'$ is a model of $\phi$ and any model of $\phi$ can be
expanded to a model of $\phi'$ by an appropriate interpretation of the additional
unary symbols. 
\end{lemma}

The proof of Lemma \ref{l:nf} first converts $\phi$ into the so-called UN-normal form (see \cite{StC13}) and then uses the standard Scott's technique \cite{Sco62}
of replacing subformulas starting with blocks of quantifiers by unary atoms built out using fresh unary symbols, and appropriately
axiomatizing the fresh unary relations. 

Lemma \ref{l:nf} allows us, when dealing with decidability/complexity issues for \UNFOEQ{}, or when considering 
the size of minimal finite models of formulas, to restrict attention to normal form sentences.

Given a structure $\str{A}$, a normal form formula $\phi$ as in (\ref{eq:nf}) and elements $a, \bar{b}$ of $A$ such that
$\str{A} \models \phi_i(a,\bar{b})$ we say that the elements of $\bar{b}$ are \emph{witnesses} for $a$ and $\phi_i$ and that
$\str{A} \restr \{a, \bar{b} \}$ is a \emph{witness structure} for $a$ and $\phi_i$.
For an element $a$ and every conjunct $\phi_i$ choose a witness structure $\str{W}_i$. Then the structure $\str{W}=\str{A} \restr \{W_1 \cup \ldots \cup W_m \}$ is called a $\phi$-\emph{witness structure} for $a$.

\subsection{Basic facts} \label{s:bf}

In \FOt{} or in \GFt{} extended by transitive relations one can enforce a transitive relation $T$ to be an equivalence (it suffices
to add conjuncts saying that $T$ is reflexive and symmetric). The same is possible, by means of a simple trick (see \cite{Kie05}), even in the variant of \GFt{} in which 
transitive relations can appear only as guards. 
It is however not possible in \UNFOTR{}. Indeed, it is not difficult to see that if $\str{A}$ is a model of an \UNFOTR{} formula $\phi$ in
which all  symbols from $\sigma_{\cDist}$ are interpreted as equivalences then another model of
$\phi$ can be constructed by taking two disjoint copies of $\str{A}$, choosing a symbol $T$ from
$\sigma_{\cDist}$, joining every element $a$ from the first copy of $\str{A}$ with its isomorphic image in the second copy by the  
$2$-type containing $T(x,y)$ as the only positive non-unary literal (in particular this $2$-type contains $\neg T(y,x)$),
and transitively closing $T$.
In this model the interpretation of $T$ is no longer an equivalence.
However:

\begin{lemma} \label{l:red}
There is a polynomial time reduction from the satisfiability  (finite satisfiability) problem for \UNFOEQ{} 
to the satisfiability  (finite satisfiability) problem for \UNFOTR{}.
\end{lemma}
\begin{proof}
Take an \UNFOEQ{} formula $\phi$, convert it into normal form  formula $\phi'$ and transform $\phi'$ into \UNFOTR{} formula $\phi''$ in the following way:
(i) replace in $\phi'$ every atom of the form $E(x,y)$ (for any variables $x, y$) by $E(x,y) \wedge E(y,x)$, (ii) add to $\phi''$ the conjunct $\forall x E(x,x)$ for every distinguished symbol $E$.
Now, any model of $\phi'$  is a model of $\phi''$; and any model of $\phi''$ can be transformed into a model of
$\phi'$ by removing all non-symmetric transitive connections. \qed
\end{proof}

The decidability and \TwoExpTime-completeness of \UNFOTR{} has been recently shown in \cite{JLM18}. Taking into consideration
that even without equivalences/transitive relations \UNFO{} is \TwoExpTime-hard we can state the following corollary.

\begin{theorem} \label{t:globalsat}
The (general) satisfiability problem for \UNFOEQ{} is \TwoExpTime-complete.
\end{theorem}

We recall that \UNFOTR{}  is contained in  \emph{the base-guarded negation fragment} with transitivity, \GNFOTR{}, in which transitive relations are
allowed only at non-guard positions, and the latter logic has been recently shown decidable and \TwoExpTime-complete by Amarilli et al. in \cite{ABBB16}. This gives an alternative argument for Thm.~\ref{t:globalsat}.
We will return to  \GNFOTR{} in Section \ref{s:gnfo}. 

As said in the Introduction both the  decidability proof for \UNFOTR{} from \cite{JLM18} and the decidability proof for \GNFOTR{} from \cite{ABBB16} strongly rely on infinite tree-like unravelings of models, and thus they give no insight into the decidability/complexity of finite satisfiability.

\medskip
Let us formulate now a simple but crucial observation on models of \UNFOEQ{} formulas.
\begin{lemma} \label{l:homomorphisms}
Let $\str{A}$ be a model of a normal form \UNFOEQ{} formula $\phi$. Let $\str{A}'$ be a structure in which all relations from 
$\sigma_{\cDist}$ are equivalences such that
\begin{enumerate}[(1)]
\item 
for every $a' \in A'$ there is a $\phi$-witness structure for $a'$ in $\str{A}'$.
\item for every tuple $a'_1, \ldots, a'_t$ (recall that $t$ is the number of variables of the $\forall$-conjunct of $\phi$) of elements
of $A'$ 
there is a homomorphism $\fh: \str{A}' \restr \{a_1', \ldots, a_t' \} \rightarrow \str{A}$ which preserves $1$-types of elements.
\end{enumerate}
Then $\str{A}' \models \phi$.

\end{lemma}
\begin{proof}
Due to (1) all elements of $\str{A}'$ have the required witness structures for all $\forall\exists$-conjuncts. It remains
to see that the $\forall$-conjunct is not violated. But since $\str{A}  \models \neg \phi_0(\fh(a_1), \ldots, \fh(a_t))$ 
and $\phi_0$ is a quantifier-free formula in which only unary atoms may be negated, it is straightforward. \qed
\end{proof}

The above observation leads in particular to a tree-like model property for \UNFOEQ{}. 
We define a $\phi$-\emph{tree-like unraveling} $\str{A}'$ of $\str{A}$ and a function $\fh:A' \rightarrow A$  in the following way. 
$\str{A}'$ is divided into levels $L_0, L_1, \ldots$. Choose an
arbitrary element $a \in A$ and put to level $L_0$ of $A'$  an element $a'$ such that $\type{\str{A}'}{a'}=\type{\str{A}}{a}$; set $\fh(a')=a$. 
Having defined $L_i$ repeat the following for every $a' \in L_i$. Choose in $\str{A}$ a $\phi$-witness structure for $\fh(a')$. Assume it consists of $\; \fh(a'), a_1, \ldots, a_s$. Add a fresh copy $a_j'$  of every $a_j$ to $L_{i+1}$, make 
$\str{A}' \restr \{a', a_1', \ldots, a_s' \}$  isomorphic to $\str{A} \restr \{\fh(a'), a_1, \ldots, a_s\}$ and set $\fh(a_i')=a_i$. Complete the definition
of $\str{A}'$ transitively closing all equivalences.

\begin{lemma} \label{l:treelike}
Let $\str{A}$ be a model of a normal form \UNFOEQ{} formula $\phi$. Let $\str{A}'$ be a $\phi$-tree-like unraveling of $\str{A}$.
Then $\str{A}' \models \phi$.
\end{lemma}
\begin{proof}
It is readily verified that $\str{A}'$ meets the properties required by Lemma \ref{l:homomorphisms}. In particular $\fh$ acts as the
required homomorphism. \qed
\end{proof}

Slightly informally, we say that a model of a normal form formula $\phi$ is \emph{tree-like} if it has a shape similar to the structure $\str{A}'$ from the above lemma, that is: 
(i) it can be divided into levels, 
(ii) every element of level $i$ has its $\phi$-witness structure completed in level $i+1$, 
(iii) $\phi$-witness structures for different elements of the same level are disjoint, 
(iv) only elements of the same witness structure may be joined by relations from $\sigma_{\cBase}$,  
(v) the only 
$\sigma_{\cDist}$-connections among elements not belonging to the same witness structure are the result of  closing transitively the equivalences in witness structures. 

\section{Small model theorem for \texorpdfstring{\UNFOtEQ}{UNFO2+EQ}}

In this section we consider \UNFOtEQ---the two-variable restriction of \UNFOEQ{}. We show the following theorem.

\begin{theorem} \label{t:smallmodeltwovars}
Every satisfiable \UNFOtEQ{} formula $\phi$ has a finite model of size bounded doubly exponentially in $|\phi|$.
\end{theorem}

As in the case of unbounded number of variables we can restrict attention to normal form formulas, which in the two-variable case 
simplify to the standard Scott-normal form for \FOt{} \cite{Sco62}:
\begin{eqnarray}
\forall {xy} \neg \phi_0({x}) \wedge \bigwedge_{i=1}^{m} \forall x \exists {y} \phi_i(x,{y}),
\label{eq:nftwovar}
\end{eqnarray}
where all $\phi_i$ are quantifier-free \UNFOt{} formulas. Without loss of generality we 
assume that $\phi$ does not use relational symbols of arity greater than $2$ (cf.~\cite{GKV97}).

Let us fix a satisfiable normal form \UNFOEQ{} formula $\phi$, and the finite relational signature $\sigma=\sigma_{\cBase} \cup \sigma_{\cDist}$ consisting 
of those symbols that appear in $\phi$. Enumerate the equivalence relation symbols as $\sigma_{\cDist}=\{E_1, \ldots, E_k \}$. Fix a 
 (not necessarily finite) $\sigma$-structure $\str{A} \models \phi$.
We will show how to build a finite model of $\phi$.

Generally, we will work in an expected way, starting from copies of some elements of $\str{A}$, adding for them
fresh witnesses  (using some patterns of connections extracted from $\str{A})$, then providing fresh witnesses for the previous 
witnesses, and so on. At some point, instead of producing new witnesses,  we
need a strategy of using only a finite number of them. It is perhaps worth explaining what are the main 
difficulties in such a kind of construction. A naive approach would be to unravel  $\str{A}$ into a tree-like structure, like 
in Lemma \ref{l:treelike}, then try to cut each branch of the tree at some point $a$ and look for witnesses for $a$
among earlier elements. The problem is when we try to reuse an element $b$ as a witness for $a$, 
and $b$ is already connected to $a$ by some equivalence relations. Then, if $a$ needs a connection to $b$
by some other equivalences, the resulting $2$-type  may become inconsistent with $\neg \phi_0$.
Another danger, similar in spirit, is that some $b$ may be needed as a witness
for several elements,  $a_1, \ldots, a_s$. 
Then
some of the $a_i$ may become connected by some equivalences which, again, may be forbidden.

It seems to be a non-trivial task to find a safe strategy of providing witnesses using only finitely many elements and
avoiding conflicts described above. 
This is why we employ a rather intricate inductive approach. We will produce substructures of the desired finite model
in which some number of equivalences are total, using patterns extracted from the corresponding substructures of the original
model. Intuitively, knowing  that an equivalence is total, we can forget about it in our construction.
Roughly speaking, our induction goes on the number of equivalence relations that are not total in the given substructures.
The constructed substructures will later become fragments of bigger and bigger substructures, which will eventually form the whole model.
To enable composing bigger substructures from smaller ones in our inductive process we will additionally keep some 
information about the intended \emph{generalized types} of elements  in form of a pattern function pointing them to
 elements in the original model.

Let us turn to the details of the proof.
Denote by $\AAA$  the set of atomic $1$-types  realized in $\str{A}$.
Note that $|\AAA|$ is bounded exponentially in $|\sigma|$ and thus also in $|\phi|$.
In this section we will use (possibly decorated) symbol $\alpha$ to denote $1$-types and $\beta$ to denote $2$-types. 

We now introduce a notion of a generalized type which stores slightly more information about an element in a structure
than its atomic $1$-type. For a set $S$ we denote by $\cP(S)$ the powerset of $S$.
\begin{definition}
A \emph{generalized type} (over $\sigma$) is a pair $(\alpha, \ff)$ where $\alpha$ is an atomic  $1$-type, and $\ff$ is an \emph{eq-visibility function}, that is a function of type $\cP(\sigma_{\cDist}) \rightarrow \cP(\AAA)$, such that, for every $\cE \subseteq \sigma_{\cDist}$
we have $\alpha \in \ff(\cE)$, and for every $\cE_1 \subseteq \cE_2 \subseteq \sigma_{\cDist}$ we have $\ff(\cE_2) \subseteq \ff(\cE_1)$.
Given a generalized type $\bar{\alpha}$ we will denote by $\bar{\alpha}.\ff$ its eq-visibility function.
We say that an element $a \in A$ \emph{realizes} a generalized type $\bar{\alpha}=(\alpha, \ff)$ in $\str{A}$,
and write $\gtype{\str{A}}{a}=\bar{\alpha}$ if 
(i) $\alpha=\type{\str{A}}{a}$, (ii)
for $\cE \subseteq \sigma_{\cDist}$, ${\bar{\alpha}.\ff}(\cE)=\{\type{\str{A}}{b}: \str{A} \models  E_i a b \text{ for all }  E_i \in \cE \}$.
We say that a generalized type $\bar{\alpha}_1=(\alpha_1, \ff_1)$ is a \emph{safe reduction} of $\bar{\alpha}_2=(\alpha_2, \ff_2)$
if $\alpha_1=\alpha_2$ and for every $\cE \subseteq \sigma_{\cDist}$ we have $\ff_1(\cE) \subseteq \ff_2(\cE)$.
We denote by $\bar{\AAA}$ the set of generalized types realized in $\str{A}$, and for $B \subseteq A$ we denote by $\bar{\AAA}[B]$ the
subset of $\bar{\AAA}$ consisting of the generalized types realized by elements of $B$.
\end{definition}

We are ready to formulate our inductive lemma.

\begin{lemma} \label{t:ind}
Let $l_0$ be a natural number $0 \le l_0 \le k$ and let $\cE_0$ be a subset of  $\sigma_{\cDist}$  of size $l_0$. 
Denote by ${\cE}_{tot}$ the set $\sigma_{\cDist} \setminus \cE_0$, and by $E^*$ the equivalence relation
$\bigcap_{E_i \in {\cE_{tot}}} E_i$.
\footnote{If $\cE_{tot}=\emptyset$ then $E^*$ is the total relation.}
Let $a_0 \in A$, let ${A}_0$ be the $E^*${-}equivalence class  of $a_0$ in $\str{A}$, and let $\str{A}_0$ be the induced substructure of $\str{A}$. 
Then there exists a finite structure $\str{B}_0$ 
and a function $\fp: B_0 \rightarrow A_0$ such that:
\begin{enumerate}[(b1)]
\item All relations from $\cE_{tot}$ are total in $\str{B}_0$. \label{btwo}
\item For every $b \in B_0$ if $\fp(b)$ has a witness $w$  for  $\phi_i (x,y)$ in $\str{A}_0$ then 
there is $w' \in B_0$ such that $\type{\str{A}_0}{\fp(b),w}=\type{\str{B}_0}{b,w'}$. \label{bthree} 
\item For every $\cE \subseteq \sigma_{\cDist}$
and $b_1, b_2 \in B_0$, if for all
 $E_i \in \cE$  $\str{B}_0 \models E_i(b_1, b_2)$ then $\gtype{\str{A}}{\fp(b_1)}.\ff(\cE)=\gtype{\str{A}}{\fp(b_2)}.\ff(\cE)$. \label{bfour}
\item For every $b \in B_0$ we have that $\gtype{\str{B}_0}{b}$ is a safe reduction of $\gtype{\str{A}}{\fp(b)}$. \label{bfive}
\item Every $2$-type realized in $\str{B}_0$ is either also realized  in $\str{A}_0$ or is obtained from a type realized
in $\str{A}_0$ by removing from it all  positive $\sigma_{\cBase}$-binary atoms and possibly some equivalence connections
and/or equalities. \label{bsix} 

\item $a_0$ is in the image of $\fp$ \label{bseven}.

\end{enumerate}
\end{lemma}

$\str{B}_0$ may be seen as a small counterpart of $\str{A}_0$ in which every element $b$
has witnesses for those $\phi_i$ for which $\fp(b)$ has a $\phi_i$-witness in $\str{A}_0$.
Intuitively, we may think that other witnesses required by $b$ are \emph{promised}
by a link to $\fp(b)$ and will be provided in further steps.

Before we prove Lemma \ref{t:ind} let us see that it indeed implies the desired finite model property from
Thm.~\ref{t:smallmodeltwovars}. To this end, take as $a_0$ an arbitrary element of $\str{A}$ and consider $l_0=k$. In this case $\cE_0=\{E_1, \ldots, E_k$\}, $\cE_{tot}=\emptyset$,
and $\str{A}_0=\str{A}$. We claim that the structure $\str{B}_0$ produced now by an application of Lemma \ref{t:ind}
is a model of $\phi$. First, Condition (b\ref{bthree}) ensures that all elements of $\str{B}_0$ have the required witnesses.
Second, (b\ref{bsix}) guarantees that for every pair of elements $b_1, b_2 \in {B}_0$ there is a homomorphism $\str{B}_0\restr \{b_1, b_2\} \rightarrow \str{A}$ preserving the $1$-types of elements; due to part (2) of Lemma \ref{l:homomorphisms} this implies that the conjunct
$\forall xy \neg\phi_0(x,y)$ is satisfied in $\str{B}_0$.

\medskip
The rest of this section is devoted to a proof of Lemma~\ref{t:ind}.
We proceed  by induction over $l_0$. Consider the base of induction, $l_0=0$.
In this case all equivalences in $\str{A}_0$ are total.
Without loss of generality assume that $|A_0|=1$. If this is not the case
just add to $\sigma_{\cDist}$ a fake symbol $E_{k+1}$ and interpret it in $\str{A}$ as the
identity relation. We take $\str{B}_0=\str{A}_0$ and
$\fp(a)=a$ for the only $a \in A_0$. Properties (b\ref{btwo})--(b\ref{bseven}) are obvious.

Let us turn to the inductive step. Assume that Thm.~\ref{t:ind} holds for some $l_0=l-1$, $0<l<k$ and let us show that it also holds for $l_0=l$. To this
end let $\cE_0$ be a subset of $\sigma_{\cDist}$ of size $l$, $a_0 \in A$ and let $\cE_{tot}$, $E^*$ and $\str{A}_0$
be as in the statement of Thm.~\ref{t:ind}. 
 Without loss of generality let us assume that $\cE_0=\{E_1, \ldots, E_l \}$. 

To build $\str{B}_0$ we first prepare some basic building blocks for our construction,
called 
\emph{components}.

\subsection{The components}

\subsubsection{Informal description and the desired properties}
A 
component is a finite structure 
having shape resembling a tree (however, not tree-like in the sense of Section \ref{s:preliminaries})
whose universe is divided into layers $L_1, \ldots, L_{l+1}$. In each layer $L_i$ we
additionally distinguish its \emph{initial part}, $L_i^{init}$.
 $L_1^{init}$ consists of a single element, called the \emph{root} of the component. The elements of layer $L_{l+1}$ are
called \emph{leaves} of the component. It may happen that some $L_i$ is empty. In such case also all layers $L_j$ for $j>i$ are 
empty, in particular there are no leaves.

We define a \emph{pattern component} for every generalized type from $\bar{\AAA}[A_0]$. 
The pattern component constructed for $\bar{\alpha}$ will be denoted $\str{C}^{\bar{\alpha}}$.
Along with the
construction of $\str{C}^{\bar{\alpha}}$  we are going to define a function $\fp$ assigning elements of $A_0$ to elements of 
$C^{\bar{\alpha}}$. Later we take some number of copies of every pattern component and join them forming the desired structure $\str{B}_0$.
The values of $\fp$ will be imported to $B_0$ from the pattern components.

Let us describe the properties which we are going to obtain during the construction of  $\str{C}^{\bar{\alpha}}$:
{\em
\begin{enumerate}[(c1)]
\item  All relations from $\cE_{tot}$ are total in $\str{C}^{\bar{\alpha}}$ . \label{ctwo}
 \item For every $c \in C^{\bar{\alpha}} \setminus L_{l+1}$ if $\, \fp(c)$ has a witness $w$  for  $\phi_i (x,y)$ in $\str{A}_0$ then there is $w' \in C^{\bar{\alpha}}$ such that $\type{\str{A}_0}{\fp(b),w}=\type{\str{C}^{\bar{\alpha}}}{c,w'}$. \label{cthree}

\item For every $\cE \subseteq \sigma_{\cDist}$
and
$c_1, c_2 \in C^{\bar{\alpha}}$, if for all 
 $E_i \in \cE$ $\str{C}^{\bar{\alpha}} \models E_i(c_1, c_2)$  then $\gtype{\str{A}}{\fp(c_1)}.\ff(\cE)=\gtype{\str{A}}{\fp(c_2)}.\ff(\cE)$. \label{cfour}
\item For every $c \in C^{\bar{\alpha}}$ we have that $\gtype{\str{C}^{\bar{\alpha}}}{c}$ is a safe reduction of $\gtype{\str{A}}{\fp(c)}$. \label{cfive} 
\item every $2$-type realized in $\str{C}^{\bar{\alpha}}$ is either a type realized also in $\str{A}_0$ or is obtained from a type realized
in $\str{A}_0$ by removing from it all  $\sigma_{\cBase}$-binary symbols and possibly some equivalences and/or equalities. \label{csix}
\item If a pair of elements is joined by a relation from $\sigma_{\cBase}$ then they belong to the same layer or to two consecutive layers. \label{cseven}

\item For $0 < i < l+1$ the elements of $L_{i}$ and $L_{i+1}$ are not joined by relation \label{ceight}
$E_i$; hence the root is not connected to any leaf by any relation from $\cE_0$.

\end{enumerate}
}

In particular, a component will satisfy almost all the properties required for $\str{B}_0$ by Thm.~\ref{t:ind}. What is missing are witnesses for leaves. 
A schematic view of a component is shown in Fig.~\ref{f:compview}.

\begin{figure} 
\begin{center}
\begin{tikzpicture}[scale=0.8]

\draw[dotted, thick] (9,0.75) -- (13.5,0.75);

\draw[dotted, thick] (11.5,3.25) -- (13.5,3.25);

\coordinate [label=center:$L_4$] (A) at ($(15,1-0.2)$); 	
\coordinate [label=center:$L_3$] (A) at ($(15,3.5-0.2)$); 	
\coordinate [label=center:$L_2$] (A) at ($(15,6-0.2)$); 	
\coordinate [label=center:$L_1$] (A) at ($(15,8.5-0.2)$); 	

%przekreslone symbole E_i
\draw[->] (1,2.5) -- (2,2);
\draw[->] (4,5) -- (5,4.5);
\draw[->] (7,7.5) -- (8,7);
\coordinate [label=center:{\color{blue} $E_3$}] (A) at ($(0.5,3)$); 
\draw (0.2,2.7) -- (0.8,3.3);
\draw (0.8,2.7) -- (0.2,3.3);
\coordinate [label=center:{\color{black} $E_2$}] (A) at ($(3.5,5.5)$); 
\draw (3.2,5.2) -- (3.8,5.8);
\draw (3.8,5.2) -- (3.2,5.8);
\coordinate [label=center:{\color{red} $E_1$}] (A) at ($(6.5,8)$); 
\draw (6.2,7.7) -- (6.8,8.3);
\draw (6.8,7.7) -- (6.2,8.3);

%inner connection
\draw[color=red, densely dashed, thick] (6.5,5.3) -- (4,3.5);
\draw[color=blue] (6.5,5.3) -- (7,3.5);
\draw[color=blue] (6.5+0.08,5.3) -- (10+0.08,3.5);
\draw[color=red, densely dashed, thick] (6.5-0.08,5.3) -- (10-0.08,3.5);
\draw[color=black, densely dotted, thick] (6.5,2.8) -- (7.5,2.8);
\draw[color=red, densely dashed, thick] (6.5,2.8) -- (7,3.5);
\draw[color=red, densely dashed, thick] (9.5,2.8) -- (10.5,2.8);

\draw[color=blue] (10,6)-- (9.5,7.8) -- (7,6);
\draw[color=blue] (9.5,7.8) -- (10.5,7.8);
\draw[color=black, densely dotted, thick] (10,8.5) -- (10.5,7.8);
\draw[color=blue] (10.5+0.08,7.8) -- (13+0.08,6);
\draw[color=black, densely dotted, thick] (10.5-0.08,7.8) -- (13-0.08,6);
\draw[color=blue] (6.5,5.3) -- (7.5,5.3) -- (7,6) -- (6.5,5.3);
\draw[color=red, densely dashed, thick] (10,5.3) -- (9.5,5.3) -- (10,6) -- (10,5.3);
\draw[color=blue] (10,6) -- (10.5,5.3);
\draw[color=blue] (13,6) -- (13.5,5.3);
\draw[color=red, densely dashed, thick] (13,6) -- (12.5,5.3);

\begin{scope}[shift={(-6,-5)}]
\draw[color=red, densely dashed, thick] (10,6)-- (9.5,7.8) -- (7,6);
\draw[color=blue] (9.5,7.8) -- (10.5,7.8);
\draw[color=black, densely dotted, thick] (10,8.5) -- (10.5,7.8);
\draw[color=red, densely dashed, thick] (10.5+0.08,7.8) -- (13+0.08,6);
\draw[color=black, densely dotted, thick] (10.5-0.08,7.8) -- (13-0.08,6);
\draw[color=blue] (6.5,5.3) -- (7.5,5.3) -- (7,6) -- (6.5,5.3);
\draw[color=black, densely dotted, thick] (10,5.3) -- (9.5,5.3) -- (10,6) -- (10,5.3);
\draw[color=blue] (10,6) -- (10.5,5.3);
\draw[color=blue] (13,6) -- (13.5,5.3);
\draw[color=black, densely dotted, thick] (13,6) -- (12.5,5.3);
\end{scope}

%additional circles
\fill[black] (10,5.3) circle (0.07);

\fill[black] (4,0.3) circle (0.07);

%triangles
\foreach \x in {0,3,6} {
\fill[black] (\x+1, 1.0) circle (0.07);
\fill[black] (\x+0.5, 0.3) circle (0.07);
\fill[black] (\x+1.5, 0.3) circle (0.07);
\draw[color=red, densely dashed, thick] (\x,0) -- (\x+1,1.5) -- (\x+2,0) -- (\x,0);}

\foreach \x in {3,6,9} {
\fill[black] (\x+1, 3.5) circle (0.07);
\fill[black] (\x+0.5, 2.8) circle (0.07);
\fill[black] (\x+1.5, 2.8) circle (0.07);
\draw[color=blue] (\x,2.5) -- (\x+1,4) -- (\x+2,2.5) -- (\x,2.5);}

\foreach \x in {6,9,12} {
\fill[black] (\x+1, 6.0) circle (0.07);
\fill[black] (\x+0.5, 5.3) circle (0.07);
\fill[black] (\x+1.5, 5.3) circle (0.07);
\draw[color=black, densely dotted, thick] (\x,5) -- (\x+1,6.5) -- (\x+2,5) -- (\x,5);}

\foreach \x in {9} {
\fill[black] (\x+1, 8.5) circle (0.07);
\fill[black] (\x+0.5, 7.8) circle (0.07);
\fill[black] (\x+1.5, 7.8) circle (0.07);
\draw[color=red, densely dashed, thick] (\x,7.5) -- (\x+1,9) -- (\x+2,7.5) -- (\x,7.5);}

\draw[color=red, densely dashed, thick] (6+0.2,5+0.1) -- (7,6.5-0.2) -- (8-0.2,5+0.1) -- (6+0.2,5+0.1);
\draw[color=red, densely dashed, thick] (3+0.2,2.5+0.1) -- (4,4-0.2) -- (5-0.2,2.5+0.1) -- (3+0.2,2.5+0.1);

\end{tikzpicture}

\caption{A component for $l=3$. Triangles correspond to subcomponents. Dashed lines represent $E_1$, dotted are used for $E_2$ and solid for $E_3$. $L_i$ and $L_{i+1}$ are not joined by $E_i$.} \label{f:compview}
\end{center}
\end{figure}
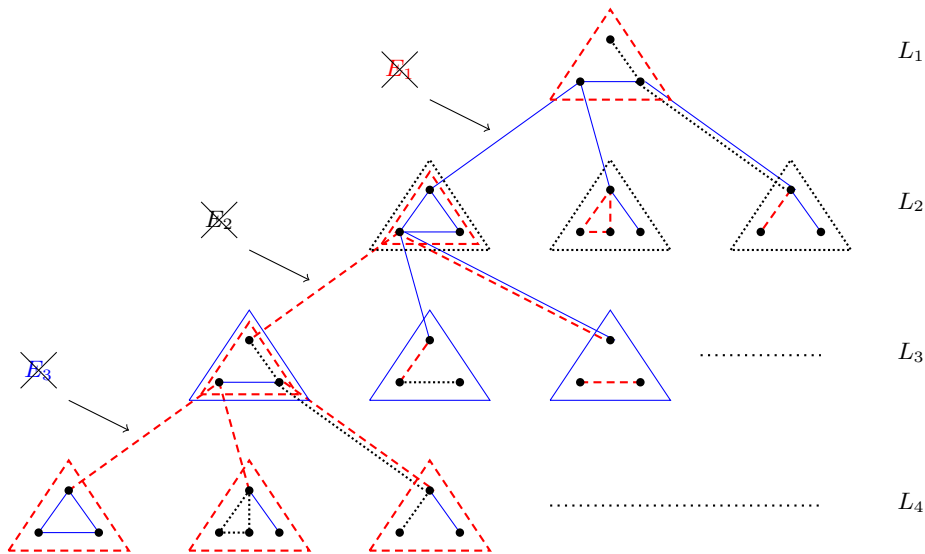
%TEST
\subsubsection{Building a pattern component.}
Let us turn to the details of construction.
Let  $\bar{\alpha}$ be a generalized type realized in $\str{A}$ by an element $r \in A_0$.
If $\bar{\alpha}$ is the type of $a_0$ then assume $r=a_0$.
We define a component $\str{C}^{\bar{\alpha}}$. To $L_1^{init}$ we put $r'$ which is a copy of $r$ (that is, $\type{\str{C}^{\bar{\alpha}}}{r'}=\type{\str{A}}{r}$),
and set $\fp(r')=r$. The element $r'$ is the root of $\str{C}^{\bar{\alpha}}$.

\medskip \noindent
{\em Step 1: Subcomponents.} Assume that we have defined $L_1, \ldots, L_{i-1}$, the initial part of $L_i$, and the structure of $\str{C}^{\bar{\alpha}}$ on $L_1 \cup \ldots \cup L_{i-1} \cup L_i^{init}$ for some $i \ge  1$. 
Assume that the values of $\fp$ on $L_1 \cup \ldots \cup L_{i-1} \cup L_i^{init}$ have also been defined. 
Let us explain how to construct the remaining part of layer $L_{i}$. 
Take any element $c \in L_{i}^{init}$. Let $a_1 = \fp(c)$. Let $A_1 \subseteq A_0$ be the $E_i$-equivalence class of $a_1$ in $\str{A}_0$
(note that $A_1$ need not be the whole $E_i$-equivalence class of $a_1$ in $\str{A}$).
Let $\cE_1=\cE_0 \setminus \{E_i \}$. Note that all relations from $\sigma_{\cDist} \setminus \cE_1$ are total in $\str{A}_1$,
and $|\cE_1|=l-1$. Thus we can use the inductive assumption for $\cE_1$, $a_1$ and $\str{A}_1$ and produce a structure $\str{B}_1$
and a function $\fp: B_1 \rightarrow A_1$, satisfying properties listed in Thm.~\ref{t:ind}. 
We put to $L_{i}\setminus L_{i}^{init}$ a copy of each element of $B_1$ besides one element $b_1$ such that $\fp(b_1)=a_1$ (such element exists due to
Condition (b\ref{bseven}) of the inductive assumption).
On the set consisting of $c$ and all the elements added in this step we define the structure isomorphic to $\str{B}_1$,
identifying $c$ with $b_1$.  We will further call such substructures 
of  components  \emph{subcomponents}.
We import the values of $\fp$ to the newly added elements of $\str{B}_1$.  
We repeat it independently for all $c \in L_{i}^{init}$. To complete the definition of  the structure on $L_1 \cup \ldots \cup L_{i}$ we just transitively  
close all the equivalences.

\medskip\noindent
{\em Step 2: Adding witnesses.}
Having defined $L_{i}$, if $i<l+1$ we now define $L_{i+1}^{init}$. 
Take any element $c \in L_{i}$. For every $1 \le j \le m$, if 
$\fp(c)$ has a witness $w \in A_0$ for $\phi_j(x,y)$ then we want to reproduce such a witness for $c$.
Let us denote $\beta=\type{\str{A}}{\fp(c),w}$. 
If $E_i(x,y) \in \beta$ then by Condition (b\ref{bthree}) of the inductive assumption $c$ has an appropriate witness in 
the subcomponent added in the previous step. 
If $E_i(x,y) \not \in \beta$ then we
add a copy $w'$ of $w$ to $L_{i+1}^{init}$, join $c$ with $w'$ by $\beta$ and set $\fp(w')=w$.  
Repeat this procedure independently for all $c \in L_{i}$. To complete the definition of  the structure on $L_1 \cup \ldots \cup L_{i} \cup L_{i+1}^{init}$ we again transitively   close all the equivalences.

The construction of the component is finished when $L_{l+1}$ is defined. 
For further purposes let us number the elements of $L_{l+1}$ of the defined pattern component $\str{C}^{\bar{\alpha}}$ as $c^{\bar{\alpha}}_1, c^{\bar{\alpha}}_2, \ldots$

Let us see that we
indeed obtain the desired properties. 

\begin{claim}
The constructed component satisfies the conditions  
below.
\begin{enumerate}[(c1)]
\item Any pair of elements belonging to the same subcomponent is connected by all relations from $\cE_{tot}$ by the inductive
assumption; every $2$-type used to connect an element of one subcomponent with its witness in another subcomponent is copied
from $\str{A}_0$, and thus it contains all relations from $\cE_{tot}$; from any element of the component one can reach every other element by connections inside subcomponents and by connections joining elements with their witnesses which means that the steps of
transitively closing $\sigma_{\cDist}$-connections will make all pairs of elements connected by all relations from $\cE_{tot}$.
\item This is explicitly taken care in \emph{Step: Adding witnesses.} A suspicious reader may be afraid that during the
step of taking transitive closure of equivalences some additional equivalences may be added to a $2$-type used to join
an element with its witness. This however cannot happen. 
It follows from 
the tree shape
of components and
 from the inductive assumption.

\item If $\cE \subseteq \cE_{tot}$ then observe that $\fp(c_1)$ and $\fp(c_2)$ are connected by all relations from $\cE$ since they both belong to $A_0$; this immediately implies the claim. 
If $\cE$ contains $E_i \not\in E_{tot}$ then by construction there is a sequence of elements $c_1=d_1, d_2, \ldots, d_{2u-1}, d_{2u}=c_2$
such that (i) $d_i$ is joined with $d_{i+1}$ by all equivalences from $\cE$, (ii) $d_{2i-1}$ and $d_{2i}$ belong to same subcomponent
(it may happen that $d_{2i-1}=d_{2i}$), and (iii) $d_{2i}$ and $d_{2i+1}$ belong to two different subcomponents and $d_{2i}$ was added as a witness for $d_{2i+1}$ or vice versa. Now,  by Condition (b\ref{bfour}) of the inductive assumption applied to subcomponents
$\gtype{\str{A}}{\fp(d_{2i-1})}.\ff(\cE)=\gtype{\str{A}}{\fp(d_{2i})}.\ff(\cE)$. By our construction $\type{\str{A}}{\fp(d_{2i}), \fp(d_{2i+1})}=
\type{\str{C}^{\alpha}}{d_{2i}, d_{2i+1}}$ and thus $\fp(d_{2i})$ and $\fp(d_{2i+1})$ are joined in $\str{A}$ by all equivalences from
$\cE$, which gives that $\gtype{\str{A}}{\fp(d_{2i-1})}.\ff(\cE)=\gtype{\str{A}}{\fp(d_{2i})}.\ff(\cE)$. It follows that
$\gtype{\str{A}}{\fp(c_1)}.\ff(\cE)=$ $\gtype{\str{A}}{\fp(c_2)}.\ff(\cE)$. 
\item The equality of $1$-types of $c_1$ and $\fp(c_1)$ follows from our choices of values of $\fp$. 
Take any $\cE \subseteq \sigma_{\cDist}$ and let $\alpha' \in \gtype{\str{C}^{\bar{\alpha}}}{c}.\ff(\cE)$. This means that there exists an element $c' \in C^{\alpha}$ of 1-type $\alpha'$ joined with $c$ by all relations from $\cE$. By (c\ref{cfour}) $\gtype{\str{A}}{\fp(c)}.\ff(\cE)=\gtype{\str{A}}{\fp(c')}.\ff(\cE)$, and since all relations
from $\cE$ are equivalences $\alpha' \in \gtype{\str{A}}{\fp(c')}.\ff(\cE)$ and thus also $\alpha' \in \gtype{\str{A}}{\fp(c)}.\ff(\cE)$.
This shows that $\gtype{\str{C}^{\bar{\alpha}}}{c}$ is a safe reduction of $\gtype{\str{A}}{\fp(c)}$.

\item  Take a $2$-type $\beta$ realized in $\str{C}^{\bar{\alpha}}$ by a pair $c_1, c_2$. If $\beta$ is realized in a subcomponent then the claim
follows by the inductive assumption applied to this substructure and the  
tree 
shape of $\str{C}^{\bar{\alpha}}$. If it joins an element of one subcomponent with its witness in another subcomponent then this
$2$-type is explicitly taken as a copy of a $2$-type from $\str{A}_0$ (cf.~also (c\ref{cthree})). Otherwise, the only positive non-unary atoms it may contain
are equivalences added in one of the steps of taking transitive closures. Let $\cE$ be the set of all equivalences belonging to $\beta$, 
and let $\alpha'$ be the $1$-type of $c_2$. By (c\ref{cfive}) $\gtype{\str{C}^{\alpha}}{c_1}$ is a safe reduction of $\gtype{\str{A}}{\fp(c_1)}$, which means that $\alpha' \in \gtype{\str{A}}{\fp(c_1)}.\ff(\cE)$. Thus there is an element $a \in A$ of $1$-type $\alpha'$ such that $\fp(c_1)$  is joined with $a$ by all equivalences from $\cE$. Observe now that $\type{\str{A}}{\fp(c_1), a}$ agrees with $\beta$ on the $1$-types it contains and contains  all equivalences which are present in $\beta$. So the claim follows.
\item Follows directly from our construction.

\item Recall that layer $L_{i+1}$ contains witnesses for elements of $L_{i}$, but each such element is joined with its witness by
a $2$-type not containing $E_i$; any path from the root to a leaf must go through all layers, thus 
for each equivalence $E_j$, $1 \le j \le l$, there is a pair of consecutive elements on this path, not joined by $E_j$.
\end{enumerate}
\end{claim}

\subsection{Joining the components}

In this step we are going to arrange a number of copies of our pattern  
components to obtain the desired structure $\str{B}_0$. 
We explicitly connect leaves of components with the roots of other components.
We do it
carefully, avoiding modifications to the internal structure
of components, which could potentially result from transitivity of relations from $\sigma_{\cDist}$. 
In particular, a pair of elements that are not connected by an equivalence $E_i \in \cE_0$ in $\str{C}$ 
will not become connected by a chain of $E_i$-connections external to $\str{C}$.

Let $max$ be the maximal number of elements in layers $L_{l+1}$ over all pattern components constructed for types from $\bar{\AAA}[A_0]$.
For every $\bar{\alpha} \in \bar{\AAA}[A_0]$ we take isomorphic copies $\str{C}^{\bar{\alpha}, g}_{i, j, \bar{\alpha'}}$ of $\str{C}^{\bar{\alpha}}$, for $g=0,1$ (we will call $g$ the \emph{color} of a component), $i=1, \ldots, max$, $j=1, \ldots, m$,  and every $\bar{\alpha}' \in \bar{\AAA}[A_0]$. This constitutes the universe of a structure $\str{B}_0^+$, together with partially defined structure (on the copies of pattern components). 
A substructure of $\str{B}_0^+$ will be later taken as $\str{B}_0$. 
We import the values of $\fp$ from $\str{C}^{\bar{\alpha}}$ to all its copies. 
Let us denote the copy of element $c^{\bar{\alpha}}_s$ from $\str{C}^{\bar{\alpha},g}_{i, j, \bar{\alpha}'}$ as $c^{\bar{\alpha},g}_{s, (i,j, \bar{\alpha}')}$.

Our strategy is now as follows: if necessary, the root of $\str{C}^{\bar{\alpha}^*,g}_{i,j, \bar{\alpha}}$ will serve as a witness of type $\bar{\alpha}^*$ for 
$\phi_j(x,y)$ and the $i$-th element
from layer $L_{l+1}$ of all copies of $\str{C}^{\bar{\alpha}}$ of color $(1{-}g)$.

Formally, for every element $c^{\bar{\alpha},g}_{s, (i',j', \bar{\alpha}')}$, 
for every $1 \le j \le m$ if
$\fp(c^{\bar{\alpha}}_s)$ has a witness $w$ for $\phi_j(x,y)$ in $\str{A}_0$ then, denoting $\bar{\alpha}^*=\type{\str{A}}{w}$
and $\beta=\type{\str{A}}{\fp(c^{\bar{\alpha}}_s), w}$, we 
join $c^{\bar{\alpha},g}_{s, (i',j', \bar{\alpha}')}$ with the root of $\str{C}^{\bar{\alpha}^*,1-g}_{s,j, \bar{\alpha}}$ using $\beta$.
See Fig.~\ref{f:joining}.
Transitively close all equivalences. 
This finishes the definition of $\str{B}_0^+$.

Finally, we choose any component $\str{C}$ whose root is mapped by $\fp$ to $a_0$ and remove from $\str{B}_0^+$ all the components which
are not accessible from $\str{C}$ in the \emph{graph of components}, formed by joining  a pair of components iff
the root of one of them serves as a witness for a leaf of another. We take the structure restricted to the remaining components as $\str{B}_0$.

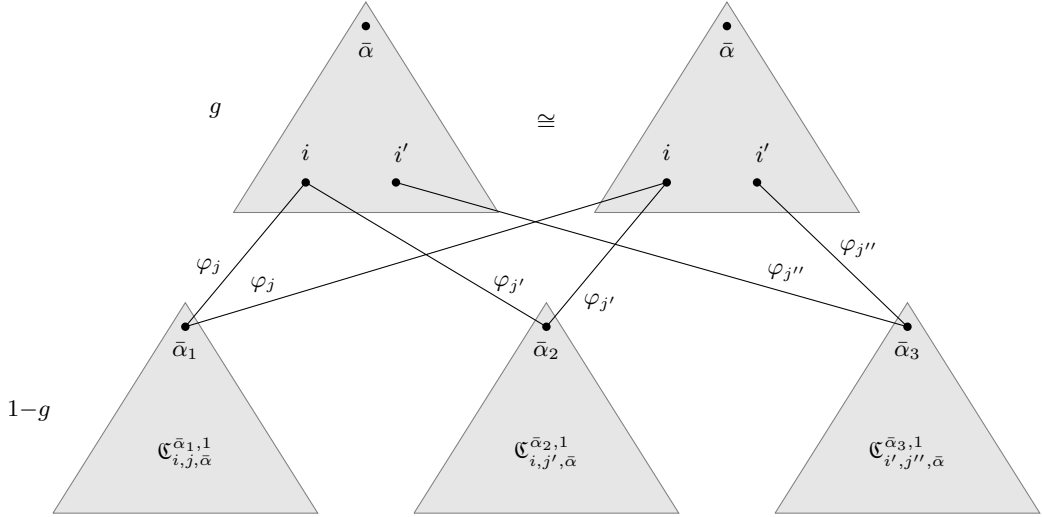
\begin{figure}  
\begin{center}
\begin{tikzpicture}[scale=0.8]

\foreach \x in {1,7,13}{
\draw[color=gray, fill=gray!20] (\x-0.2,-0.5) -- (\x+4+0.2,-0.5) -- (\x+2,3) -- (\x-0.2,-0.5);
\fill[black] (\x+2,2.6) circle (0.07);
}
\foreach \x in {4,10}{
\draw[color=gray, fill=gray!20] (\x-0.2,4.5) -- (\x+4+0.2,4.5) -- (\x+2,8) -- (\x-0.2,4.5);
\fill[black] (\x+2,7.6) circle (0.07);
}

\coordinate [label=center:$g$] (A) at ($(3.5,6.2)$); 
\coordinate [label=center:$1{-}g$] (A) at ($(0.4,1.2)$); 

\coordinate [label=center:$\bar{\alpha}_1$] (A) at ($(3,2.2)$); 
\coordinate [label=center:$\bar{\alpha}_2$] (A) at ($(9,2.2)$);
\coordinate [label=center:$\bar{\alpha}_3$] (A) at ($(15,2.2)$);

\coordinate [label=center:$\str{C}^{\bar{\alpha}_1,1}_{i,j,\bar{\alpha}}$] (A) at ($(3,0.5)$); 
\coordinate [label=center:$\str{C}^{\bar{\alpha}_2,1}_{i,j',\bar{\alpha}}$] (A) at ($(9,0.5)$);
\coordinate [label=center:$\str{C}^{\bar{\alpha}_3,1}_{i',j'',\bar{\alpha}}$] (A) at ($(15,0.5)$);

\coordinate [label=center:$\bar{\alpha}$] (A) at ($(6,7.2)$); 
\coordinate [label=center:$\bar{\alpha}$] (A) at ($(12,7.2)$);

\coordinate [label=center:$i$] (A) at ($(5,5.5)$); 
\coordinate [label=center:$i'$] (A) at ($(6.6,5.5)$); 
\coordinate [label=center:$i$] (A) at ($(11,5.5)$); 
\coordinate [label=center:$i'$] (A) at ($(12.6,5.5)$);

\coordinate [label=center:$\phi_j$] (A) at ($(3.4,3.6)$); 
\coordinate [label=center:$\phi_j$] (A) at ($(4.3,3.3)$); 
\coordinate [label=center:$\phi_{j'}$] (A) at ($(9.9,3.0)$);
\coordinate [label=center:$\phi_{j'}$] (A) at ($(8.4,3.3)$);
\coordinate [label=center:$\phi_{j''}$] (A) at ($(13,3.5)$);
\coordinate [label=center:$\phi_{j''}$] (A) at ($(14.2,3.9)$);

\coordinate [label=center:$\cong$] (A) at ($(9,6)$); 

\foreach \x in {5,6.5, 11,12.5}{
\fill[black] (\x,5) circle (0.07);
}

\draw (5,5) -- (3,2.6) -- (11,5) -- (9,2.6) -- (5,5);
\draw (6.5,5) -- (15,2.6) -- (12.5,5);

\end{tikzpicture}
\caption{Joining the components.} \label{f:joining}
\end{center}
\end{figure}

\subsection{Correctness of the construction}

Let us first observe the following basic fact.

\begin{claim} \label{c:klejm}
The process of joining the 
components does not change the previously defined internal structure of any component.
\end{claim}
\begin{proof}
Potential changes could result only from  closing transitively the equivalences which join leaves of some components with their witnesses---the roots of other components. Recall that by Condition (c\ref{ceight}) the root of a component is not connected by any equivalence to any leaf of this component and note first that this condition cannot be violated in the step of joining components. This is guaranteed by our strategy requiring
leaves of components of color $g$ to take as witnesses the roots of  components of color $(1{-}g)$, for $g=0,1$.

Consider now any $E_i \in \cE_0$ and elements $c_1, c_2$ belonging to the same component $\str{C}$. Assume that 
$\str{C} \not\models E_i(c_1, c_2)$, but $\str{B}_0 \models E_i(c_1, c_2)$. This means
that during the process of providing witnesses for leaves, an  $E_i$-path joining $c_1$ and $c_2$ was formed. Take such a path. Due to
Condition (c\ref{ceight}) such a path cannot enter a component through a leaf and leave it through the root. Thus, 
without loss of generality, we can assume that  it is of the form $c_1, d_1^{\sss out}$, $r_1, d_2^{\sss in}$, $d_2^{\sss out}, r_2, \ldots$, $d_{s-1}^{\sss in}, d_{s-1}^{\sss out}, r_{s-1}, d_{s}^{\sss in}, c_2$, where 
the two elements of every pair ($c_1, d_1^{\sss out}$), ($d_2^{\sss in}$, $d_2^{\sss out}$), $\ldots$, ($d_{s}^{\sss in}, c_2$) are members of the same component, all $d^{(\cdot)}_i$ are leaves, and each $r_i$ is the root of a component  used as a witness for $d_i^{\sss out}$ and $d_{i+1}^{\sss in}$. See Fig.~\ref{f:path}.
Recalling our strategy, allowing a root to be used as a witness only for copies of the same leaf from some pattern component, 
we see that the components containing  pairs ($d_i^{\sss in}$, $d_{i}^{\sss out}$), for $i=2, \ldots, s-1$ and the component $\str{C}$ containing $c_1$ and $c_2$
are isomorphic to one another. 
Mapping isomorphically the $E_i$-edges joining $d_{i}^{\sss in}$ with $d_{i}^{\sss out}$ to $\str{C}$  we see
that $c_1$ and $c_2$ were already connected by  $E_i$ in $\str{C}$. Contradiction. \qed 
\end{proof}

\begin{figure}  
\begin{center}
\begin{tikzpicture}[scale=0.8]

\foreach \x in {1,7,13}{
\draw[color=gray, fill=gray!20] (\x-0.2,4.5) -- (\x+4+0.2,4.5) -- (\x+2,8) -- (\x-0.2,4.5);

\draw[color=gray, fill=gray!20] (\x-0.2,-0.5) -- (\x+4+0.2,-0.5) -- (\x+2,3) -- (\x-0.2,-0.5);

\fill[black] (\x+2,2.6) circle (0.07);
}

\coordinate [label=center:$\str{C}$] (A) at ($(3,7)$); 
\fill[black] (2.5,6) circle (0.07);
\coordinate [label=center:$c_1$] (A) at ($(2.5,6.35)$); 
\fill[black] (3.5,6) circle (0.07);
\coordinate [label=center:$c_2$] (A) at ($(3.5,6.35)$); 
\coordinate [label=center:$d_1^{\sss out}$] (A) at ($(1.7,5.4)$); 
\coordinate [label=center:$d_4^{\sss in}$] (A) at ($(4.4,5.4)$); 
\coordinate [label=center:$d_2^{\sss in}$] (A) at ($(8,5.53)$); 
\coordinate [label=center:$d_2^{\sss out}$] (A) at ($(9.1,5.5)$); 
\coordinate [label=center:$d_3^{\sss in}$] (A) at ($(15,5.53)$); 
\coordinate [label=center:$d_3^{\sss out}$] (A) at ($(16.1,5.5)$);
\coordinate [label=center:$E_i$] (A) at ($(2,3.5)$);

\coordinate [label=center:$r_1$] (A) at ($(3,2.2)$); 
\coordinate [label=center:$r_2$] (A) at ($(9,2.2)$);
\coordinate [label=center:$r_3$] (A) at ($(15,2.2)$);

\coordinate [label=center:$\cong$] (A) at ($(6,6)$); 

\coordinate [label=center:$\cong$] (A) at ($(12,6)$); 

\foreach \x in {2,3,4,8,9,10,14,15,16}{
\fill[black] (\x,5) circle (0.07);
}

%path:
\draw (2.5,6) -- (2,5) -- (3,2.6) -- (8,5) -- (9,5) -- (9,2.6) -- (15,5) -- (16,5) -- (15,2.6) -- (4,5) -- (3.5,6);

\draw[dotted] (2, 5) -- (3,5) -- (4,5);  

\end{tikzpicture}
\caption{An $E_i$-path joining $c_1$ and $c_2$} \label{f:path}
\end{center}
\end{figure}
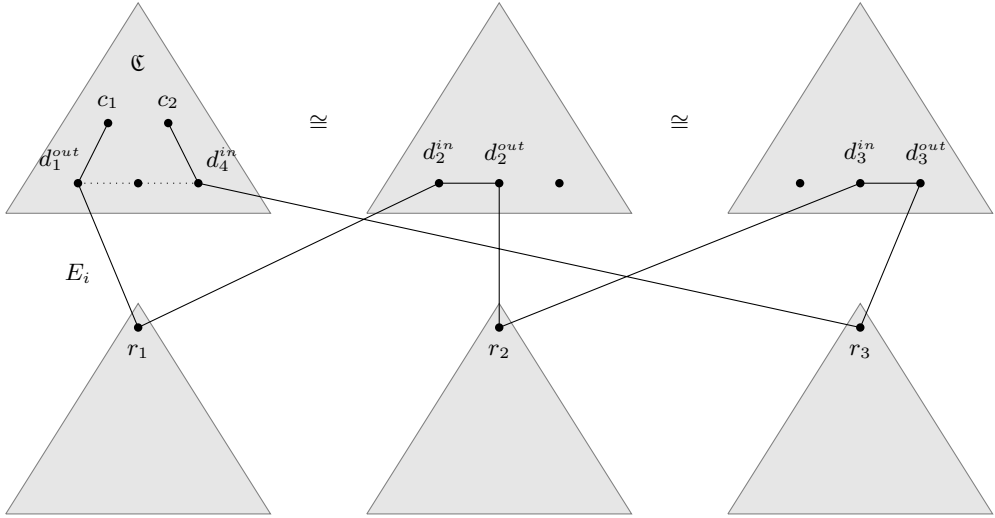

Now, Conditions  (b\ref{btwo})--(b\ref{bseven}) can be shown using
arguments similar to the ones used in the proofs of the above claim and (c\ref{ctwo})--(c\ref{cseven}). 
\subsection{Proof of conditions (b\ref{btwo})--(b\ref{bseven})}
\begin{enumerate}[(b1)]
\item This is taken care in the last step of the construction, when we take $\str{B_0}$ as a "connected" substructure
of $\str{B}_0^+$. Recall that by (c\ref{ctwo}) all relations from $\cE_{tot}$ are total in components, and that every
$2$-type joining a leaf with its witness contains all equivalences from $\cE_{tot}$.
\item All elements of layers $L_1, \ldots, L_{l}$ of any component have witnesses in their component.
For witnesses from the last layer $L_{l+1}$ of every component we take care in the step of joining the components.
The argument that the $2$-types declared during the step of providing witnesses will not be modified during the step
of taking transitive closures of equivalences is similar to the one in the proof of Claim \ref{c:klejm}.

\item The proof is very similar to the proof of Condition (c\ref{cfour}) for components, but has a slight modification due to the joining procedure.
By (c\ref{ceight}) there exists $g\in\{0,1\}$ such that there is no $E$-path joining $b_1$ and $b_2$ which uses a direct connection between the root of a component of color $g$ and a leaf of a component of color $1-g$. Firstly, by Claim \ref{c:klejm}, isomorphic components observation, if both $b_1$ and $b_2$ are in some components of colors $g$, we can assume that they are is the same component. Now we can use Claim \ref{c:klejm}-like projection argument to find for all $E\in\cE$ such $E$-paths joining $b_1$ with $b_2$ that they all use the same set of edges created during the joining step. Now we can proceed as in (c\ref{cfour}) using (c\ref{cfour}) and the fact, that for neighbouring $b,b'$ belonging to different components $\gtype{\str{A}}{(\fp(b)}.\ff(\cE)=\gtype{\str{A}}{(\fp(b')}.\ff(\cE)$.

\item Follows from (b\ref{cfour}) exactly as (c\ref{cfive}) follows from (c\ref{cfour}).
\item The proof is analogous to the proof of (c\ref{csix}). Again, this time the role of basic substructures is played by components.
\item This condition is taken care explicitly when a component for the generalized type of $a_0$ is constructed: $a_0$ becomes then
the value of $\fp$ for the root of the component.

\end{enumerate}

This finishes the proof of Lemma \ref{t:ind} and thus also the proof of the finite model property for \UNFOtEQ.

\subsection{Size of models and complexity of \texorpdfstring{\UNFOtEQ}{UNFO2+EQ}}

 To complete the proof of Thm.~\ref{t:smallmodeltwovars} we need to
estimate the size of finite models produced by our construction.
This can be done by formulating
a recurrence relation for $T_l$---an upper bound 
on the size of structure $\str{B}_0$ constructed in the proof of
Lemma~\ref{t:ind} for $l_0=l$. Note that the size of our final model
is bounded by $T_{k+1}$. (We use $T_{k+1}$ rather than $T_{k}$ since in the base of induction we may need to add an auxiliary equivalence.)

Clearly $T_0=1$. The size of a single basic substructure used
in the case $l_0=l+1$ is bounded by $T_l$. In $L_1$ there is one
such substructure. Each of its elements produces at most $m$
elements in $L_2^{init}$, each of them expanding to a basic substructure.
Thus $|L_2| \le T_lm T_l$. Inductively, $|L_i| \le T_l^{i} m ^{i-1}$. The values of estimates of $|L_i|$ form
a geometric series, whose sum (=an estimate on the size of a component) can be bounded by $T_l^{l+2}m^{l+1}$.
Denoting the number of generalized types realized in $\str{A}_0$ by $K$ the number of components is
$K \cdot 2 \cdot (T_l^{l+1}m^{l}) \cdot m \cdot K$.
Thus we get
$$T_{l+1} \le 2K^2 T^{l+1}m^{l+1} \cdot T_l^{l+2}m^{l+1}=2K^2m^{2l+2}T_l^{2l+3}.$$
Since $K$ is bounded doubly exponentially and $m, l$---polynomially in $|\phi|$, the solution of this recurrence relation
allows to estimate $|T_{k+1}|$  doubly exponentially in $|\phi|$.

\medskip
We conclude this section with the following observation.
\begin{theorem}
The satisfiability (= finite satisfiability) problem for \UNFOtEQ{} is \TwoExpTime-complete.
\end{theorem}

\begin{proof}
The upper bound follows from the finite model property and the upper bound for general satisfiability problem for \UNFOEQ{} formulated in  Thm.~\ref{t:globalsat}.
The lower bound can be shown by a routine adaptation of the proof of a \TwoExpTime-lower
bound for the two-variable guarded fragment with two equivalence relations from \cite{Kie05}.
A simple inspection of the properties  needed to be expressed in that proof shows that they need only unary negations. 

We also remark that a similar construction can be used to show that the doubly exponential upper bound on the size of models of satisfiable
\UNFOtEQ{} formulas is essentially optimal, that is \UNFOtEQ{} it is possible to enforce models of at least doubly exponential size. \qed
\end{proof}

\section{Small model theorem for full \UNFOEQ} \label{s:maineq}

In this section we explain how to extend the small model theorem from the previous section to the case in which the
number of variables is unbounded. The general approach is similar: given a pattern model we 
inductively rebuild it into a finite one. The first difference is that this inductive construction will be preceded
by a pre-processing step producing from an arbitrary pattern model a model which has regular tree-like shape. 
Assuming such regularity 
will allow not only for a simpler description of the main construction,
but, more importantly, for a simpler argument that the finite model we build satisfies part (2) from Lemma \ref{l:homomorphisms}. 

Secondly, the number of layers
of 
components we are going to construct needs to be increased with respect to the two-variable case. 
This
time we not only require that the root of a component is not connected with any leaf by any (non-total) equivalence---we
 use a stronger property that
in particular
implies that there is
no path from the root to a leaf built out of equivalence connections, on which the equivalences alternate less
than $t$ times (recall that $t$ is the number of variables in the $\forall$-conjunct). 

The third difference we
want to point out concerns the construction of witness structures. In the two-variable case a witness structure 
for a given element $a$ and $\phi_i$ consisted of $a$ and just one additional element and in the inductive process
it was created at once. Now such witness structures are bigger.  Moreover,  for simplicity, 
we will deal with full  $\phi$-witness structures rather
than with witness structures for various $\phi_i$ separately. 
Given a tree-like model we will allow ourselves to speak about \emph{the}  $\phi$-witness structure
for an element, meaning the 
witness structure consisting of this element and its all children, even if, accidentally, some other $\phi$-witness structures
for this element exist. 
In a single inductive step  usually only
some parts of $\phi$-witness structures are created (the parts in which the appropriate equivalences are total)
and the remaining
parts are completed in the higher levels of induction. 
Such fragments of $\phi$-witness structures considered in a single inductive step
will be referred to as \emph{partial $\phi$-witness structures}.

 Finally, generalized types from Section 3 will no longer be sufficient for our purposes. The role of a type
of an element will be played this time by the isomorphism type of the subtree rooted at the pattern 
of this element.

\subsection{Regular tree-like models}

\begin{lemma} \label{l:regular} Every satisfiable \UNFO{} normal form  formula  $\phi$ has
a tree-like model $\str{A}\models\phi$ with doubly exponentially many (with respect to $|\phi|$)  non-isomorphic subtrees.
\end{lemma}

The proof starts from a tree-like model guaranteed by Lemma \ref{l:treelike}. Then, roughly speaking, 
some patterns which could possibly be extended to  substructures falsifying the $\forall$-conjunct of $\phi$
are defined. A node of a tree-like model is assigned a \emph{declaration}, that is 
the list of such patterns which do not appear in its subtree. We choose one node for every realized declaration and
build a regular tree-like model out of copies of the chosen elements and their $\phi$-witness structures.
As the number of possible declarations is bounded doubly exponentially,  the claim follows. We omit the details of the proof,
referring the reader to the proof of an analogous fact for a more general scenario involving arbitrary transitive 
relations rather than equivalences, see \cite{DK18c}.

\subsection{Main theorem}

We are now  ready to show the main result of this paper.

\begin{theorem} \label{t:maintr}
Every satisfiable \UNFOEQ{} formula $\phi$ has a model of size bounded doubly exponentially in $|\phi|$.
\end{theorem}

Let us fix a satisfiable normal form \UNFOEQ{} formula $\phi$, and the finite relational signature $\sigma=\sigma_{\cBase} \cup \sigma_{\cDist}$ consisting of all symbols appearing in $\phi$.  Enumerate the equivalences as $\sigma_{\cDist}=\{E_1, \ldots, E_k \}$. Fix a regular tree-like $\sigma$-structure $\str{A} \models \phi$ with at most doubly exponentially many non-isomorphic subtrees, which exists due to Lemma \ref{l:regular}. We show how to build a finite model of $\phi$.
We mimic the inductive approach and the main steps of a finite model construction for $\phi$ from the previous section. However, the details are more complicated.

Recall that in the two-variable case, we built our finite structure together with a function $\fp$ whose purpose was to assign to elements of the new model elements of the original model of similar generalized types. Intuitively, in the current construction the role of generalized types of elements will be   played by the isomorphism types of subtrees of $\str{A}$.

An important property of the substructures created during our inductive process is that they admit some partial homomorphisms to the pattern tree-like model $\str{A}$ which restricted to (partial) witness structures act as isomorphisms into the corresponding parts of the $\phi$-witness structures in $\str{A}$.
We impose that every homomorphism respects the
this condition using directly the structure of $\str{A}$.
To this end we introduce further fresh (non-equivalence) binary symbols $W^i$ whose purpose is to relate elements to their witnesses.
We number the elements of the $\phi$-witness structures in $\str{A}$ arbitrarily (recall that each element is a member of its own $\phi$-witness structure)
and interpret $W^i$ in $\str{A}$ so that
for each $a,b\in A$, $\str{A}\models W^iab$ iff $b$ is the  $i$-th element of the $\phi$-witness structure for $a$ (from now, for short, we refer to the element $b$ satisfying $W^iab$ as the $i$-th witness for $a$).
We do this in such a way that if two subtrees of $\str{A}$ were isomorphic before interpreting the $W^i$ then they still are after such expansion.
Now, if we mark $b$ as the $i$-th witness for $a$ during the construction (that is set $\str{A}'\models W^iab$), then for any 
homomorphism $\fh$ we have $\str{A}\models W^i\fh(a)\fh(b)$.

To shorten notation we will denote by $[a]_E$ the $E$-equivalence class of an element $a$ (the structure will be clear from the context).
We denote by $\str{A}_a$ the subtree rooted at $a$ (from now on such subtrees will be considered only in  $\str{A}$).
We state the counterpart of Lemma \ref{t:ind} as follows. 

\begin{lemma} \label{l:finiteeq} Let $\cE_0\subseteq\sigma_{\cDist}$, $\cE_{tot}=\sigma_{\cDist}\backslash\cE_0$, $E^*=\bigcap_{E_i\in\cE_{tot}}E_i$, $a_0\in A$, $\str{A}_0$ be the induced substructure of $\str{A}$ on $A_{a_0}\cap[a_0]_{E^*}$. Then there exists a finite structure $\str{A}'_0$, an element (called the \origin{} of $\str{A}_0'$) $a_0'\in A'_0$ and a function $\fp:A'_0\to A_0$ such that:
	\begin{enumerate}[(d1)]
		\item  $E^*$ is total on $\str{A}_0'$.\label{done}
		
		\item  $\fp(a_0')=a_0$.\label{dtwo}
		
		\item
	
		For each $a'\in A'_0$ and each $i$, if the $i$-th witness for $\fp(a')$
		lies in $A_0$ (that is $\str{A}_0\models\exists y\; W^i\fp(a')y$) then there exists a unique element $b'\in A_0'$ such that $\str{A}_0'\models W^i a'b'$. Otherwise there exists no such element. Denote $W_{a'}=\{b':\exists i\;\str{A}_0'\models W^i a'b'\}$ and for a tuple $\bar{a}$ let $W_{\bar{a}}=\bigcup_{a\in\bar{a}}W_{a}$. \label{dtwohalf}
		
		\item 
		For each $\bar{a}\subseteq A_0'$ satisfying $|\bar{a}|\leq t$ there exists a homomorphism $\fh:\str{W}_{\bar{a}}\to\str{A}_0$ such that for each $a\in\bar{a}$ we have $\str{A}_{\fp(a)}\cong\str{A}_{\fh(a)}$ and $\fh\restr W_a$ is an isomorphism (onto its image).
	
		Moreover, if $a_0'\in\bar{a}$ then we can choose $\fh$ so that $\fh(a_0')=a_0$.\label{dthree}
		
		\item
	
		For each $a\in A_0'$ we have $\str{W}_a\cong\str{W}\restr A_0$ where $\str{W}$ is the $\phi$-witness structure for $\fp(a)$.
	(Note that, by the definition of the $W^i$, each such isomorphism sends $a$ to $\fp(a)$.)
		\label{dfour}

	\end{enumerate}
\end{lemma}

The proof goes by induction on $l=|\mathcal{E}_0|$. 
The base of induction, $l=0$, can be treated as in the two-variable case.
For the inductive step, suppose that theorem holds for $l-1$. We show that it holds for $l$.
Without loss of generality let $\cE_0=\{E_1,\ldots,E_l\}$. The rest of the proof is presented in Sections \ref{s:patterns}--\ref{s:correctenss}.

\subsection{Pattern components}  \label{s:patterns}
 In the two variable case we created a single type of a building block 
for every generalized type realized in substructure $\str{A}_0$ of the original model.
Now we create one type of a building block for every isomorphism type of a subtree rooted at a node of $\str{A}_0$.
We denote by $\GGG[A_0]$ the set of such isomorphism types. 
 Let $\gamma_{a_0}$ be the type of $\str{A}_{a_0}$. 

 Take $\gamma \in \GGG[A_0]$ and the root $a \in A_0$ of a subtree of type $\gamma$. If $\gamma=\gamma_{a_0}$, take $a=a_0$.
We explain how to construct a finite \emph{pattern component} 
$\str{C}^{\gamma}$. The main steps of this construction are similar to the ones in the two-variable case. 
This time the component is divided into $l(2t+1)+1$ \emph{layers} $L_1,\ldots,L_{l(2t+1)+1}$.  The first $l(2t+1)$ of them are called
\emph{inner layers} while the last one is called the \emph{interface layer}. We start the construction of an inner layer $L_i$ by defining its initial part, $L_i^{init}$, and then expand it to a full layer. 
The interface layer $L_{l(2t+1)+1}$ has no internal division but, for convenience, is sometimes
referred to as $L_{l(2t+1)+1}^{init}$.
The elements of $L_{l(2t+1)}$ are called \emph{leaves} and 
the elements of $L_{l(2t+1)+1}$ are called \emph{interface elements}. For technical reasons, the bottom of a component is organized in a slightly different way than in the two-variable case, where leaves were in the last layer and there was no notion of an interface layer.
In the current construction, the interface elements will be later identified with the roots of some other components.

 $\str{C}^\gamma$ will have a shape resembling a tree, with  structures obtained by the inductive assumption as nodes. All elements of the inner layers of $\str{C}^\gamma$ will have appropriate partial $\phi$-witness structures provided.

 We remark that, in contrast to the two-variable case, during the process of building a pattern component we do not yet apply the
transitive closure to the equivalence relations.
 Taking the transitive closures would not affect the correctness of the construction, but not doing this at this point
will allow us for a simpler presentation of the correctness proof.
 Given a pattern component $\str{C}$ we will sometimes denote by $\str{C}_+$ the structure obtained from $\str{C}$ by
applying the appropriate transitive closures.
 The crucial property we want to enforce is that the root of $\str{C}^\gamma$ will be  \emph{far} from its leaves in the following sense. 
Denote by $G_l(\str{S})$,  for a $\sigma$-structure $\str{S}$, the  Gaifman graph of the structure obtained by removing from $\str{S}$ the equivalences $E_{l+1}, \ldots, E_k$. Then
there will be no connected induced subgraph of $G_l(\str{C}^\gamma_+)$ of size $t$ containing an element of one of the first $l$ layers
and, simultaneously, an element of one of the last $l$ inner layers of $\str{C}^\gamma$.

 We set $L_1^{init}=\{a'\}$ to consist of a copy of element $a$, i.e., we set $\type{\str{C}^{{\gamma}}}{a'}:=\type{\str{A}_0}{a}$. Put $\fp(a')=a$.  We call $a'$ the \emph{root} of  $\str{C^{\gamma}}$.

\smallskip\noindent
 \emph{Construction of a layer}. Suppose we have defined layers $L_1,\ldots,L_{i-1}$ and $L_i^{init}$, $1\leq i\leq l(2t+1)$, and the structure
and the values of $\fp$  
on $L_1 \cup \ldots \cup L_{i-1} \cup L_{i}^{init}$. 
We now explain how to define $L_i$ and $L_{i+1}^{init}$. Let $s=1+(i-1 \mod l)$.

\smallskip\noindent
\emph{Step 1: Subcomponents}. Take any element $c\in L_i^{init}$. 
From the inductive assumption we have a structure $\str{B}_0$ with $E^*\cap E_s$ total on it, its \origin{} $b_0\in B_0$ and a function $\fp_c:B_0\to A_{\fp(c)}\cap[\fp(c)]_{E^*\cap E_s}\subseteq A_0$
with $\fp_c(b_0)=\fp(c)$. The substructures obtained owing to the inductive assumption are called \emph{subcomponents}.
We identify $b_0$ with $c$, add isomorphically $\str{B}_0$ to $L_i$, and extend function $\fp$ so that $\fp\restr B_0=\fp_c$. We do this independently for all $c\in L_i^{init}$. 

\smallskip\noindent
\emph{Step 2: Providing witnesses}. 
This step is slightly different compared to its two-variable counterpart. For $i<l(2t+1)+1$ we now define $L_{i+1}^{init}$. 
Take $c\in L_i$. Let $\str{W}$ be the $\phi$-witness structure for $\fp(c)$ in $\str{A}$. 
Let $\str{F}$ be the restriction of $\str{W}$ to $[\fp(c)]_{E^*\cap E_s}$.
Let $\str{F}'$ be the  isomorphic copy of $\str{F}$  created for $c$ in
the subcomponent 
$\str{B}_0$ 
built in Step 1 that contains $c$
($\str{F}'$ exists due to (d\ref{dfour})).
Let $\str{E} = \str{W} \restr [\fp(c)]_{E^*}$. 
We add $F''$---a copy of $E\setminus F$ to $L_{i+1}^{init}$, and isomorphically copy the structure of $\str{E}$ to $F'\cup F''$
identifying $F'$ with $F$.
See Fig.~\ref{f:witnesseseq}.
Note that this operation is consistent with the previously defined  structure on $\str{F}'$. 
 The structure on $F'\cup F''$ will be the structure $\str{W}_c$ in $\str{C}^{\gamma}$ and then in $\str{A}_0'$.
  We define $\fp\restr F''$ in a natural way, for each element  $b \in F''$ choosing as the value of $\fp(b)$ the  isomorphic counterpart of $b$ in $E \setminus F$.
We repeat this step independently for all for all $c\in L_i$.

\begin{figure}
	\centering
	\begin{tikzpicture}[scale=1.6*0.3]
	\draw (6,8) -- (10,8) -- (11,4) -- (10,2) -- (6,2) -- (5,4) -- (6,8);
	\draw[fill=gray!20] (5,4) -- (6,8) -- (7,4) -- (5,4);
	\draw[fill=gray!10] (5,4) -- (0,4) -- (6,8) -- (5,4);
	
	\fill[black] (4,4) circle (0.12);
	\fill[black] (3,4) circle (0.12);
	\fill[black] (2,4) circle (0.12);
	\fill[black] (1,4) circle (0.12);
	\fill[black] (6,8) circle (0.12);
	\fill[black] (5.5,4) circle (0.12);
	\fill[black] (6.5,4) circle (0.12);

	%drzewo
	\draw[fill=gray!20] (19,4) -- (20,8) -- (21,4) -- (19,4);
	\draw[fill=gray!10] (19,4) -- (14,4) -- (20,8) -- (19,4);
	\draw (20,8) -- (26,4);
	\draw (26,4) -- (19,4);
	\draw[dotted] (26,4)--(28,0);
	\draw[dotted] (14,4)--(12,0);

	\fill[black] (18,4) circle (0.12);
	\fill[black] (17,4) circle (0.12);
	\fill[black] (16,4) circle (0.12);
	\fill[black] (15,4) circle (0.12);
	\fill[black] (20,8) circle (0.12);
	\fill[black] (19.5,4) circle (0.12);
	\fill[black] (20.5,4) circle (0.12);

	%wygieta strzlka
	\draw[bend left=20, ->, very thin] (6.1,8.1) to (19.9,8.1);

	\draw [
	%thick,
	decoration={
		brace,
		mirror,
		raise=0.5cm
	},
	decorate
	] (0.6,4.8) -- (4.4,4.8);
	
	\coordinate [label=center:$L_{i+1}^{init}$] (A) at ($(2.5,1)$);
	\draw[->] (2.5,3.3) -- (2.5,2);

	\draw [
	%thick,
	decoration={
		brace,
		mirror,
		raise=0.5cm
	},
	decorate
	] (5,2.5) -- (11,2.5);
	
	\coordinate [label=center:in $L_i$] (A) at ($(8,-0.3)$);

	%Napisy

	\coordinate [label=center:${E}{\setminus}F$] (A) at ($(16.3,4.6)$); 
	
	\coordinate [label=center:${F''}$] (A) at ($(2.5,4.6)$); 
	
	\coordinate [label=center:$F$] (A) at ($(20,4.6)$); 
	
	\coordinate [label=center:$F'$] (A) at ($(6,4.6)$); 
	
	\coordinate [label=center:$\str{B}_0$] (A) at ($(9,4)$); 
	
	\coordinate [label=center:$\str{W}$] (A) at ($(24,6.5)$);

	\coordinate [label=center:$c$] (A) at ($(6,8.7)$); 
	\coordinate [label=center:$\fp(c)$] (A) at ($(20,8.7)$); 
	
	\coordinate [label=center:$\str{A}_{\fp(c)}$] (A) at ($(25,2.5)$);

	\end{tikzpicture}
	\caption{Providing witnesses.}
	\label{f:witnesseseq}%
\end{figure}
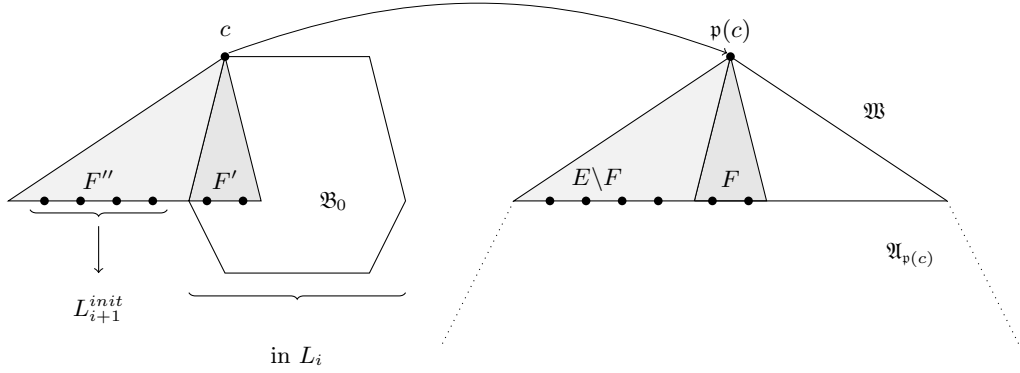

\smallskip
When the interface layer, $L_{l(2t+1)+1}^{init}$ ($= L_{l(2t+1)+1}$),  is created the construction of $\str{C}^{\gamma}$ is completed.

\subsection{Joining the components}

 As in the case of \UNFOtEQ, this step consists in joining some leaves with some roots of components. 
To deal with the additional `moreover' part of condition (d\ref{dthree}) we will simply define $a_0'$ in such a way that it will not be used as a witness for any leaf.
As promised above we create pattern components for all types from $\GGG[A_0]$. 
Let $max$ be the maximal number of 
interface elements
over all pattern components.
For each $\str{C}^{\gamma}$ we number its interface elements.
We create components $\str{C}^{\gamma,g}_{i,\gamma'}$  for all $\gamma, \gamma' \in \GGG[A_0]$, $g \in \{0,1\}$ ($g$ is often called a \emph{color}), $1 \le i \le max$, 
as isomorphic copies of $\str{C}^{\gamma}$. We also create an additional component 
$\str{C}^{\gamma_{a_0},0}_{\bot, \bot}$ as a copy of $\str{C}^{\gamma_{a_0}}$, and define $a_0'$ to be its root.

 For each $\gamma$, $g$ consider components of the form $\str{C}^{\gamma,g}_{\cdotp,\cdotp}$. Perform the following
procedure for each $i$---the number of an interface element. 
Let $b$ be the $i$-th interface element of any such component, let  $\gamma'$ be the type of $\str{A}_{\fp(b)}$.
Identify the $i$-th interface elements of all $\str{C}^{\gamma,g}_{\cdotp,\cdotp}$
with the root $c_0$ of $\str{C}^{\gamma',1-g}_{i,\gamma}$. 
 Note that the values of $\fp(c_0)$ and $\fp(b)$ (the latter equals to the value of $\fp$ on the $i$-th interface element in all the $\str{C}^{\gamma,g}_{\cdotp,\cdotp}$) may differ. However, by construction,   $\str{A}_{\fp(b)}\cong\str{A}_{\fp(c_0)}$
(in particular, the 1-types of $b$ and $c_0$ match). For the element $c^*$ obtained in this identification step we 
define $\fp(c^*)=\fp(c_0)$.

 Finally, we take as $\str{A}_0^0$ the 
structure restricted to the components accessible 
in the graph of components from $\str{C}^{\gamma_{a_0}, 0}_{\bot, \bot}$. 
The graph of
components $G^{comp}$ is formed by joining  a pair of components iff
we identified the root of one of them with an interface element of the other. 

 We now define $\str{A}_0'$ as $\str{A}_0^0$ with transitively closed equivalences
and set the root of $\str{C}_{\bot,\bot}^{\gamma_{a_0},0}$ to be its \origin.
Recall that in the structure $\str{A}^0_0$ we, exceptionally, do not transitively close $\sigma_{\cDist}$-connections, and thus allow the interpretations of the symbols from $\sigma_{\cDist}$ 
not to be transitive  (we will keep using superscript $0$ for auxiliary structures of this kind).

\subsection{Correctness of the construction} \label{s:correctenss}
 Now we proceed to the proof that $\str{A}_0'$ satisfies Conditions (d\ref{done})--(d\ref{dfour}).

\smallskip\noindent
(d\ref{done}) 
After taking the transitive closures, $E^*$ is total on each pattern component. Thus, by the definition of the graph of components $G^{comp}$, $E^*$ is total on $\str{A}_0'$. 

\smallskip\noindent
(d\ref{dtwo})  Follows directly from the  definition of $L_1^{init}$  in $\str{C}^{\gamma_{a_0}}_{\bot, \bot}$ and the fact that $C^{\gamma_{a_0}}_{\bot, \bot}\subseteq A_0'$.

\smallskip\noindent
(d\ref{dtwohalf}) 
The interpretations of the $W^i$ are defined in the step of providing witnesses where, implicitly, we take
care of this condition for every element $a'$ of the inner layers by extending the fragment of the partial $\phi$-witness structure
for $a'$
created on the previous level of induction
by a copy of a further fragment of the same pattern  $\phi$-witness structure. The identifications of elements during the step of joining the components  
do not spoil the required property and cause that it holds for all elements of $\str{A}_0'$.

\smallskip\noindent
(d\ref{dthree})  This is the key part of our argumentation. 
 For simplicity, let us ignore the `moreover' part of this condition for some time. We will explain how to take care of it near the end of this proof.
Now we find a homomorphism $\fh$ such that 
$\str{A}_{\fp(a)}\cong\str{A}_{\fh(a)}$ for all $a\in\bar{a}$ (we say that such a homomorphism has \emph{the subtree isomorphism property}).
Later we will show that its restrictions to the substructures $\str{W}_a$  are indeed isomorphisms. 
 The proof consists of several homomorphic reductions performed in order to show that we
can restrict attention to a structure built as a component but twice as high.

\begin{figure}  
\begin{center}
	\begin{tikzpicture}[scale=0.75]

%lower triangles
\foreach \x in {1,7,13}{
\draw[color=gray, fill=gray!20] (\x-0.2,-0.5) -- (\x+4+0.2,-0.5) -- (\x+2,3) -- (\x-0.2,-0.5);
\fill[black] (\x+2,2.6) circle (0.07);
}

%upper triangles
\foreach \x in {4,10}{
\draw[color=gray, fill=gray!20] (\x-0.2,4.5) -- (\x+4+0.2,4.5) -- (\x+2,8) -- (\x-0.2,4.5);
\draw[color=gray, fill=gray!20] (\x+0.1,5) -- (\x+4-0.1,5);

\coordinate [label=center:$\cong$] (A) at ($(9,6.3)$); 
\coordinate [label=center:$\str{F}_0'$] (A) at ($(2.2,4.8)$); 
\coordinate [label=center:$g$] (A) at ($(3.5,6.2)$); 
\coordinate [label=center:$1{-}g$] (A) at ($(0.4,1.2)$); 
\coordinate [label=center:$\str{C}^{\gamma}$] (A) at ($(6,6.9)$); 
\coordinate [label=center:$\pi$] (A) at ($(9,7.7)$);

\draw[->] (2.6, 4.5) -- (3,4);

%elements
\foreach \x in {4.5,6,7.5 , 10.5,12, 13.5}{
\fill[black] (\x,4.75) circle (0.07);
}

\coordinate [label=center:$b_1$] (A) at ($(4.8,5.25)$); 
\coordinate [label=center:$b_2$] (A) at ($(6,5.25)$); 
\coordinate [label=center:$b_3$] (A) at ($(7.3,5.25)$); 

\coordinate [label=center:$c_0$] (A) at ($(3,2.2)$);

%identifications
\draw[dashed] (4.5,4.75) -- (3,2.6) -- (10.5,4.75);
\draw[dashed] (6,4.75) -- (9,2.6) -- (12,4.75);
\draw[dashed] (7.5,4.75) -- (15,2.6) -- (13.5,4.75);

%snakes
\draw[decorate, decoration={snake, segment length=10, amplitude=1}] (4.9,6.2) -- (7.1,6.2);
\draw[decorate, decoration={snake, segment length=10, amplitude=1}] (10.9,6.2) -- (13.1,6.2);
\draw[decorate, decoration={snake, segment length=10, amplitude=1}] (1.9,1.2) -- (4.1,1.2);
\draw[decorate, decoration={snake, segment length=10, amplitude=1}] (7.9,1.2) -- (10.1,1.2);
\draw[decorate, decoration={snake, segment length=10, amplitude=1}] (13.9,1.2) -- (16.1,1.2);

\draw [decorate,decoration={brace,amplitude=10pt},rotate=0] (14,6.2) -- (17,1.2) node [black,midway, xshift=17, yshift=5] {$\;\;\;\supseteq \str{W}_{\bar{a}}$};

%witness structures
\draw[dotted] (18,-1) -- (0,-1) -- (6,9) -- (9,4) -- (15,4) -- (18,-1);
\draw[loosely dotted] (18.3,-1.2) -- (-0.3, -1.2) -- (6,9.3) -- (12,9.3) -- (18.3,-1.2); %EK17 nowe: E_0'

%wygieta strzlka
\draw[bend right=10, ->,  dashed] (10.5,7.2) to (7.5,7.2);

}

\end{tikzpicture}
	\caption{Joining the components and Reductions 1 and 2. Elements connected by dashed lines are identified. } \label{f:reductions}
	\end{center}
\end{figure}

\medskip\noindent
\emph{Reduction 0}.  Take $\bar{a}\subseteq A_0'$, $|\bar{a}|\leq t$. 
 Observe that for each $a\in\bar{a}$ the structure $\str{W}_a$ is connected in $G_l(\str{A}_0'\restr W_{\bar{a}})$ (recall the definition of Gaifman graph $G_l(\str{S})$ and the interpretation of the symbols $W^i$).
Let $\str{W}_{\bar{a}_1},\ldots,\str{W}_{\bar{a}_K}$ be the connected components of $\str{W}_{\bar{a}}$ in $G_l(\str{A}_0'\restr W_{\bar{a}})$.
 If we have homomorphisms $\fh_i:\str{W}_{\bar{a}_i}\to\str{A}_0$, it is sufficient to put $\fh=\bigcup\fh_i$ as the desired homomorphism, since $E^*$ is total on $\str{A}_0$ and for $a\in\bar{a}_i$ we also have $\str{A}_{\fh(a)}=\str{A}_{\fh_i(a)}\cong\str{A}_{\fp(a)}$.
So \emph{we can restrict attention to tuples $\bar{a}$ with $\str{W}_{\bar{a}}$ connected in the above sense}.

\medskip\noindent
\emph{Reduction 1}.  The key fact is that, informally, $\str{W}_{\bar{a}}$ is contained `on a
boundary of two colors'. That is, there exists $g\in\{0,1\}$ such that 
removing all the connections between leaves of color $1-g$ and roots of color $g$ (in other words: any connections between
elements of $L_{l(2t+1)}$ and elements of $L_{l(2t+1)+1}$ in components of color $1-g$) does not remove any connection among the elements of $\str{W}_{\bar{a}}$. This property follows from the fact that each subcomponent `kills' one of the $E_i$, therefore, by the arrangement of subcomponents in a component, a connected $\str{W}_{\bar{a}}$ may be spread over a limited number of layers and the number of layers in a component is chosen high enough so the above property holds.

 Reformulating, let $\str{D}_0^0$ be a structure obtained from $\str{A}_0^0$ by removing all direct connections between roots of color $g$ and leaves of color $1-g$ and $\str{D}_0'$ its minimal extension in which  equivalences are transitively closed. We have just proved that the inclusion map $\iota:\str{W}_{\bar{a}}\to\str{D}_0'$ is a homomorphism, and since for all $a\in\bar{a}$, $\str{A_{\fp(a)}}=\str{A_{\fp(\iota(a))}}$,  \emph{we can restrict attention to a tuple $\bar{a}$ for which $\str{W}_{\bar{a}}$ is connected and search for a homomorphism $\str{W}_{\bar{a}}\to\str{A}_0$ treating $\str{W}_{\bar{a}}$ as a substructure of $\str{D}_0'$}.

\medskip\noindent
\emph{Reduction 2.}  Consider the shape of a connected fragment of the graph of components $G^{comp}$ with connections between leaves of color $g$ and roots of color $1-g$ removed.
 Observe that there is at most one type $\gamma$ of components of color $g$, chosen in the previous reduction, containing some element of $\str{W}_{\bar{a}}$ and all elements of $\str{W}_{\bar{a}}$ of color $1-g$ are contained in components of the form $\str{C}^{\cdotp,1-g}_{\cdotp,\gamma}$.  See~Fig.~\ref{f:reductions}. 
Now we can naturally `project' all the elements of $\str{W}_{\bar{a}}$ of color $g$ on one chosen component $\str{C}^{\gamma}$ of type $\gamma$ and color $g$. Call this projection $\pi$. Then we remove from $\str{D}_0^0$ all components of color $g$ other than $\str{C}^{\gamma}$ and all components of color $1-g$ of form other than $\str{C}^{\cdotp, 1-g}_{\cdotp,\gamma}$ obtaining a structure $\str{F}_0^0$. Let $\str{F}_0'$ be created by closing transitively all equivalences in $\str{F}_0^0$. We claim that $\pi$ is a homomorphism from $\str{W}_{\bar{a}}$ to $\str{F}_0'$. Indeed such projection can be applied to paths in $\str{D}_0^0$ to get corresponding paths in $\str{F}_0^0$. Since for all $a\in\bar{a}$ we have $\str{A}_{\fp(a)}=\str{A}_{\fp(\pi(a))}$,  \emph{we may restrict attention to a tuple $\bar{a}$ for which $\str{W}_{\bar{a}}$ is connected and search for a homomorphism $\str{W}_{\bar{a}}\to\str{A}_0$ treating $\str{W}_{\bar{a}}$ as a substructure of $\str{F}_0'$}. 

\medskip\noindent
\emph{Essential homomorphism construction}.
 By the construction of $\str{A}_0'$ we can see that $\str{F}_0^0$ can be considered as a component of height $2l(2t+1)$ and such component can be viewed, as a tree $\tau$ whose nodes are  subcomponents: we make subcomponent $\str{B}$ a parent of $\str{B}'$ iff $\str{B}'$ contains a witness for an element of $\str{B}$. We will build a homomorphism $\fh:\str{W}_{\bar{a}}\to\str{A}_{a_0}$ inductively using a bottom-up approach on tree $\tau$. 
For a subcomponent $\str{B}$ denote by $B^{\wedge}$ the union of the domains of all the subcomponents belonging to the subtree of $\tau$ rooted at $\str{B}$.

 Since we might have cut some connections between an element and some of its witnesses during Reduction 1, we define for each $a\in F_0'$ the surviving
part $\str{V}_a$ of $\str{W}_a$ by $\str{V}_a=\str{F}_0'\restr V_a$ where $V_a=\{b:\exists i\;\str{F}_0'\models W^iab\}$. For a tuple $\bar{b}$ denote $V_{\bar{b}}=\bigcup_{b\in\bar{b}} V_b$ and $\str{V}_{\bar{b}}=\str{F}_0'\restr V_{\bar{b}}$.
Note that $V_a\subseteq W_a$,
and generally, this inclusion may be strict,
but for all $a\in\bar{a}$ we have $\str{V}_a = \str{W}_a$,
and thus,
in particular, the claim below finishes the proof of the currently
considered part of (d\ref{dthree}), that is the
proof of
the existence of a homomorphism satisfying the subtree isomorphism property.

 Returning to the shape of $\str{F}_0^0$, it consists of some subcomponents arranged into tree $\tau$ glued together by the structure on the surviving parts. Note that all such building blocks (that is both the subcomponents and the surviving parts of the partial witness structures) are transitively closed. Moreover, by the tree structure of $\tau$, if some elements of such a  building block are connected by some atom in $\str{F}_0'$, then they already have been connected by the same atom in $\str{F}_0^0$, therefore  the identity map from $\str{F}_0^0$ to $\str{F}_0'$ acts as an isomorphism when restricted to such a building block. 

Recall that due to the expansion of the structure defined before the statement of Lemma \ref{l:finiteeq}, all homomorphisms $\str{A}_0'\to\str{A}_0$ respect the numbering of witnesses. This property will be particularly important in the proof of the following claim.

\begin{claim}\label{c:joiningeq} 
	 For every subcomponent $\str{B}_0\in\tau$ with \origin{} $b_0$, and for every  $\bar{a}\subseteq B_0^{\wedge}$,
	$|\bar{a}|\leq t$, there exists a homomorphism $\fh:\str{V}_{\bar{a}}\to\str{A}_{\fp(b_0)}\restr[\fp(b_0)]_{E^*}$ such that for all $a\in\bar{a}$
	we have $\str{A}_{\fh(a)}\cong\str{A}_{\fp(a)}$, 
	and if $b_0\in\bar{a}$ then $\fh(b_0)=\fp(b_0)$.
\end{claim}

\begin{proof}
 Bottom-up induction on subtrees of $\tau$.	
	
\medskip\noindent
\emph{Base of induction}.  In this case $\str{W}_{\bar{a}}\subseteq\str{B}_0$ and the claim follows from the inductive assumption of Lemma \ref{l:finiteeq}.

	\begin{figure}  
		\begin{center}
		\begin{tikzpicture}[scale=0.8]
		
		%elipse upper
		\draw (10,10) ellipse (50pt and 25pt);
		\fill[black] (10,10.3) circle (0.07);
		\coordinate [label=center:$a_1$] (A) at ($(10,10.6)$); 
		\fill[black] (9,9.5) circle (0.07);
		\coordinate [label=center:$c_1$] (A) at ($(9,9.8)$); 
		\fill[black] (11,9.5) circle (0.07);
		\coordinate [label=center:$c_2$] (A) at ($(11,9.8)$); 
		\coordinate [label=center:$\str{B}_0$] (A) at ($(10,9.5)$); 
		
		%wintess str
		\draw (9,9.5) -- (8,8.2);
		\draw (11,9.5) -- (12,8.2);
		\draw (1.5,4) -- (2,5);
		\draw (4.5,4) -- (4,5);
		\fill[black] (2,5) circle (0.07);
		\fill[black] (4,5) circle (0.07);
		\coordinate [label=center:$\fh_0(c_1)$] (A) at ($(2.1,5.5)$); 
		\coordinate [label=center:$\fh_0(c_2)$] (A) at ($(3.9,5.5)$);

		%elipsse left
		\draw (8,8) ellipse (50pt and 25pt);
		\fill[black] (8,8.2) circle (0.07);
		\coordinate [label=center:$b_1$] (A) at ($(7.8,8.6)$); 
		\fill[black] (7,7.5) circle (0.07);
		\coordinate [label=center:$a_2$] (A) at ($(7,7.8)$); 
		\fill[black] (9,7.5) circle (0.07);
		\coordinate [label=center:$a_3$] (A) at ($(9,7.8)$); 
		\coordinate [label=center:${B}_1^{\wedge}$] (A) at ($(8,7.6)$); 
		
		%ellipse right
		\draw (12,8) ellipse (50pt and 25pt);
		\fill[black] (12,8.2) circle (0.07);
		\coordinate [label=center:$b_2$] (A) at ($(12.2,8.6)$); 
		\fill[black] (13,7.5) circle (0.07);
		\coordinate [label=center:$a_4$] (A) at ($(13,7.8)$); 
		\coordinate [label=center:${B}_2^{\wedge}$] (A) at ($(12,7.6)$); 
		
		%triangle 1
		\draw (7,3) -- (8,5.5) -- (9,3);
		\fill[black] (8,5) circle (0.07);
		\coordinate [label=center:$\fh_1(b_1)$] (A) at ($(9.1,5)$); 
		\fill[black] (7.5,3.5) circle (0.07);
		\fill[black] (8.5,3.5) circle (0.07);
		
		%triangle 2
		\draw (11,2.5) -- (12,5.0) -- (13,2.5) ;
		\fill[black] (12,4.5) circle (0.07);
		\coordinate [label=center:$\fh_2(b_2)$] (A) at ($(13.1,4.5)$); 
		\fill[black] (12.5,3) circle (0.07);

		%big triangle
		\draw (-0.5,2) -- (3,9) -- (6.5,2);
		\fill[black] (3,7.5) circle (0.07);
		\coordinate [label=center:$\fh_0(a_1)$] (A) at ($(3,7)$);

		%triangle 1'
		\draw (0.5,2) -- (1.5,4.5) -- (2.5,2) ;
		\fill[black] (1.5,4) circle (0.07);
		\fill[black] (1,2.5) circle (0.07);
		\fill[black] (2,2.5) circle (0.07);
		\coordinate [label=center:${\sss \cong \str{A}_{\fh_1(b_1)}}$] (A) at ($(1.5,1.8)$); 
		\coordinate [label=center:$b_1'$] (A) at ($(0.9,4)$);

		%triangle 2'
		\draw (3.5,2) -- (4.5,4.5) -- (5.5,2) ;
		\fill[black] (4.5,4) circle (0.07);
		\fill[black] (5,2.5) circle (0.07);
		\coordinate [label=center:${\sss \cong \str{A}_{\fh_2(b_2)}}$] (A) at ($(4.5,1.8)$); 
		\coordinate [label=center:$b_2'$] (A) at ($(3.9,4)$);

		%tiny triangle 1
		
		%tiny triangle 2
		
		%bent arrows
		
		\coordinate [label=center:$\fh_0$] (A) at ($(5.5,10)$); 
		\draw[bend right=15, ->, dashed] (7.5,10) to (4,8.5);	
		
		\draw[bend right=5, ->, dashed] (8,6.9) to (8,5.7);	
		\coordinate [label=center:$\fh_1$] (A) at ($(8.5,6.4)$); 
		
		\draw[bend right=5, ->, dashed] (12,6.9) to (12,5.2);	
		\coordinate [label=center:$\fh_2$] (A) at ($(12.5,6.05)$);

		\draw[bend right=5, ->, dashed] (7.5,5.2) to (2,4.2);	
		
		\draw[bend right=3, ->, dashed] (11.5,4.5) to (5,3.7);	
		
		\end{tikzpicture}
		\caption{Joining homomorphisms} \label{f:joininghom}
		\end{center}
	\end{figure}
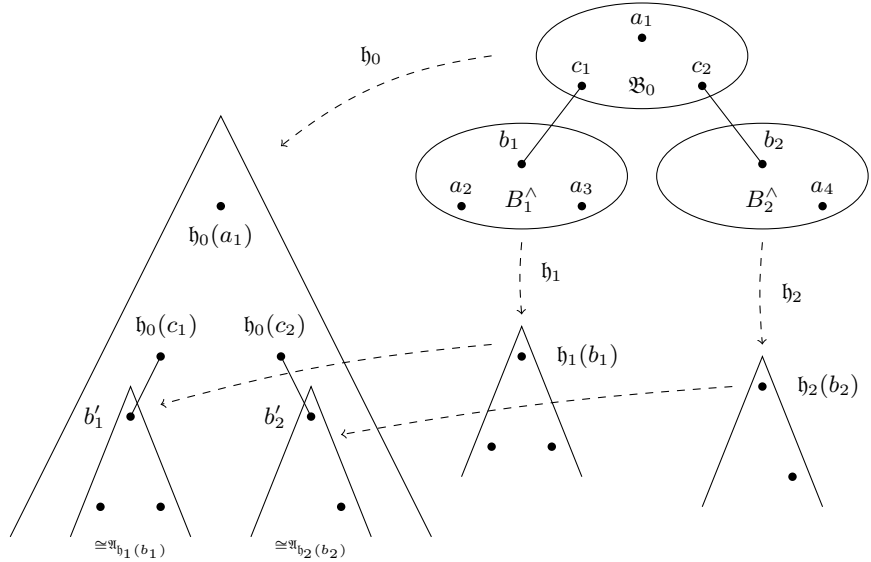
	
	\medskip\noindent
	\emph{Inductive step.}  Let $\str{B}_1,\ldots,\str{B}_K$ be the list of those children of  $\str{B}_0$ in $\tau$ for which  $B_i^{\wedge}$ contains some elements of $\bar{a}$; denote by $b_i$ the root of $\str{B}_i$ and let $c_i\in B_0$ be such that $b_i$ is a witness chosen by $c_i$ in the step of providing witnesses/joining the components.
	If $K=1$ and $\bar{a}\subseteq B_1^{\wedge}$ the thesis follows from the inductive assumption of this claim.

	 Otherwise, by the inductive assumption of this claim applied to $(\bar{a}\cap B_i^{\wedge})b_i$ we have homomorphisms $\fh_i:\str{V}_{(\bar{a}\cap B_i^{\wedge})b_i}\to \str{A}_{\fp(b_i)}$ satisfying $\fp(b_i)=\fh_i(b_i)$ and from the inductive assumption of Lemma \ref{l:finiteeq} a homomorphism $\fh_0:\str{V}_{(\bar{a}\cap B_0)c_1\ldots c_K}\restr B_0\to\str{A}_{\fp(b_0)}$. 
	 We extend the latter in the only possible way to $\fh_0^*$ defined on the whole $\str{V}_{(\bar{a}\cap B_0)c_1\ldots c_K}$: for each $a\in\bar{a}$ and $c\in V_a\setminus B_0$ (by construction $\str{V}_a\models W^iac$ for some $i$) we set $\fh(c)$ to be the only element satisfying $\str{A}_0\models W^i\fh(a)\fh(c)$ (such an element exists since $\str{A}_{\fh(a)}\cong\str{A}_{\fp(a)}$---in particular the $\phi$-witness structures of $\fh(a)$ and $\fp(a)$
	are isomorphic). 
	Note that the sizes of the tuples used to build the homomorphisms $\fh_i$ are bounded by $t$, as required.

	 Using regularity of $\str{A}$, homomorphisms $\fh_0^*,\fh_1,\ldots,\fh_K$ can be joined together\linebreak
into $\fh:\str{V}_{\bar{a}b_1\ldots b_Kc_1\ldots c_K}\to \str{A}_{\fp(b_0)}$ (see Fig.~\ref{f:joininghom}).
	 In order to attach
	$\fh_i$ to $\fh_0^*$ we define $\fh_i^*$. Let $j$ be such that $b_i$ is the $j$-th witness for $c_i$ and let $b_i'$ be the $j$-th witness for $\fh_0(c_i)$ (it exists by $\str{A}_{\fh_0(c_i)}\cong\str{A}_{\fp(c_i)}$). Then we have $\str{A}_{b_i'}\cong\str{A}_{\fp(b_i)}$ since both $b_i'$ and $\fp(b_i)$ are the $j$-th witnesses of some elements of $\str{A}$ being the roots of isomorphic subtrees. Thus, composing $\fh_i$ with such an isomorphism  gives a homomorphism $\fh_i^*:\str{V}_{(\bar{a}\cap B_i^{\wedge})b_i} \to \str{A}_{b_i'}$ with $\fh_i^*(b_i)=b_i'$. Finally we set $\fh=\bigcup\fh_i^*$.  Note that $\fh$ is well defined (the value of $\fh$ on each of the $b_i$ has been defined twice).  
	
	 For each $a\in\dom\fh_i$ ($=\dom \fh_i^*$, when $i>0$)
	we have 
	$\str{A}_{\fh(a)}=\str{A}_{\fh_i^*(a)}\cong\str{A}_{\fh_i(a)}$($\cong\str{A}_{\fp(a)}$,
	by the inductive assumptions of this claim and Lemma \ref{l:finiteeq}). Since $\bar{a}\subseteq \dom\fh_0\cup\bigcup_{i>0}\dom\fh_i^*$, we 
	get that for each $a\in\bar{a}$ we have $\str{A}_{\fp(a)}\cong\str{A}_{\fh(a)}$.
	
	 Recalling the tree structure on $\tau$ we can conclude that $\fh$ is a homomorphism. We give an idea of the proof of this property. 	Consider an $E_u$-path in $\str{F}_0^0$ connecting two elements of $\str{V}_{\bar{a}b_1\ldots b_Kc_1\ldots c_K}$. 
	We show that the images of these two elements are connected by an $E_u$-path in $\str{A}$.
	Using the tree shape of $\str{F}_0^0$, we can split it into parts contained in $B_0$ or some of the $B_i^{\wedge}$,
	and parts contained in some of the 
	$V_d$ for $d\in(\bar{a}\cap B_0)c_1\ldots c_K$
	(with the
	splitting points belonging to 
	$V_{(\bar{a}\cap B_0)c_1\ldots c_K}$). 
	For the former
	type of connections, use the fact that $\fh_0,\fh_1^*,\ldots,\fh_K^*$ are homomorphisms. For the latter, observe that
	$\fh_0^*$ sends 
	$V_d$
	into the corresponding part of an isomorphic copy of the pattern $\phi$-witness
	structure from $\str{A}$ used to define the structure on $\str{F}_0^0\restr V_d$.
		Similarly a non-transitive relation in $\str{F}_0^0$ may connect elements contained in $B_0$ or one of the $B_i^{\wedge}$, or one of the 	$V_d$ for $d\in(\bar{a}\cap B_0)c_1\ldots c_K$,
	and the argument as above shows that it is preserved by $\fh$.

	 It follows from the construction that $\fh$ has the following property: if $b_0\in\bar{a}$ then $\fh(b_0)=\fh_0(b_0)=\fp(b_0)$. To finish the proof of the inductive step, we restrict $\fh$ to $V_{\bar{a}}$. \qed
\end{proof}
 Now we prove the additional property required for $\fh$ by (d\ref{dthree}), namely that $\fh \restr W_a$ is an isomorphism.
 Observe that $\fh$ injectively moves $\str{W}_a$ into the corresponding part of the $\phi$-witness structure for $\fh(a)$ which is isomorphic to the corresponding part of the $\phi$-witness structure for $\fp(a)$ by the subtree isomorphism property. Therefore, since the structure on $W_a$ (prior to taking the transitive closure) was copied from the latter, the inverse of $\fh\restr W_a$ is a homomorphism and therefore $\fh\restr W_a$ is an isomorphism.  

 To prove the `moreover' part of (d\ref{dthree}), it suffices to observe that if $a_0'\in\bar{a}$ then in Reduction 1 we have that $g=0$ and in Reduction 2 we have that $\gamma=\gamma_{a_0}$. We choose $\str{C}^{\gamma}=\str{C}^{\gamma,0}_{\bot,\bot}$. 
This way the application of the Reductions does not move $a_0'$. The claim follows from the fact that $\fh(a_0')=\fp(a_0')=a_0$.

\smallskip\noindent
(d\ref{dfour})
 Apply (d\ref{dthree}) to a tuple consisting of just $a$ to obtain an isomorphism $\fh:\str{W}_a\to\str{A}_0\restr\fh(W_a)$
and then apply an isomorphism between $\str{A}_{\fh(a)}$ and $\str{A}_{\fp(a)}$.

\medskip

This finishes the proof of Lemma \ref{t:ind}.
 Let us show how this lemma implies the finite model property for \UNFOEQ{}.
Take $\cE_0=\sigma_{\cDist}$, let $a_0$ be the root of $\str{A}$. We apply Lemma \ref{l:finiteeq} and get a finite structure $\str{A}_0'$ and a function $\fp:A_0'\to A_0$. Note, that $\str{A}_0=\str{A}$. Let us see that $\str{A}_0'$ satisfies the conditions of Lemma \ref{l:homomorphisms}. Indeed, (1) follows from (d\ref{dfour}). Condition (2) follows from (d\ref{dthree}). So $\str{A}_0'\models\phi$.

\subsection{Size of models and complexity}

Now we show, that the size of $\str{A}_0'$ is bounded doubly exponentially in $|\phi|$. We calculate a recurrence
relation on $ M_l$---an upper bound on the size of the
structure created in the $l$-th step of induction. We are interested in an estimate for $ M_{k+1}$. 

 Let $n=|\phi|$.
Consider the $l$-th
induction step. 
The size of each subcomponent is bounded by $ M_{l-1}$. Consider one component. 
Layer $L_1$ consists of at most $ M_{l-1}$ elements, each of them creates at most $n$ elements in layer $L_2^{init}$, which jointly create at most $ M_{l-1}\cdotp n\cdotp M_{l-1}$ elements in layer $L_2$ and inductively at most $ M_{l-1}^in^{i-1}$ elements in layer $L_i$. 
So each component has at most $ M_{l-1}^{l(2t+1)+2}n^{l(2t+1)+2}$ elements. 
Counting the components, we get an estimate $$ M_l=  M_{l-1}^{8n^2}\cdotp n^{8n^2}\cdotp(|\GGG[A]|\cdotp2\cdotp M_{l-1}^{8n^2}\cdotp |\GGG[A]|+1).$$
Solving this recurrence 
relation we get $$ M_{k+1}\leq (|\GGG[A]|^2\cdotp4\cdotp n^{8n^2})^{(16n^2)^{n+1}},$$ which is doubly exponential in $n$. 

The finite model property and   Thm.~\ref{t:globalsat} allow us to conclude.
\begin{theorem}
The finite satisfiability problem for \UNFOEQ{} is \TwoExpTime-complete.
\end{theorem}

\section{Towards guarded negation with equivalences} \label{s:gnfo}
We observe now that our small model construction can be adapted 
for a slightly bigger logic.
The \emph{guarded negation fragment} of first-order logic, \GNFO{}, is defined in \cite{BtCS15} by the following grammar:
$$\phi=R(\bar{x}) \mid x=y \mid \phi \wedge \phi \mid \phi \vee \phi \mid \exists x \phi \mid \gamma(\bar{x}, \bar{y}) \wedge \neg \phi(\bar{y}),$$
where $\gamma$ is an atomic formula. Since equality statements of the form $x=x$ can be used as guards, \GNFO{} may be viewed as an extension of \UNFO{}.
However, the satisfiability problem for \GNFO{} with equivalences is undecidable. It follows from the fact that even the two-variable guarded fragment,
which is contained in \GNFO{}, becomes undecidable when extended by equivalences \cite{Kie05}.

To regain decidability we consider
the \emph{base-guarded negation fragment with equivalences},\linebreak
\GNFOEQ{}, analogous to the base-guarded negation fragment with transitive relations,\linebreak 
\GNFOTR, investigated in \cite{ABBB16}.
In these variants all guards must belong to $\sigma_{\cBase}$, and all symbols from $\sigma_{\cDist}$ must be interpreted as equivalences/transitive relations. Recall that the general satisfiability problem for \GNFOTR{} was shown decidable in \cite{ABBB16}, and as 
explained in Section \ref{s:bf} this implies decidability of the general satisfiability problem for \GNFOEQ{}. In this paper we do not solve the
finite satisfiability problem for full \GNFOEQ{}. We, however, do solve this problem for its one-dimensional restriction.

We say that a first-order formula is \emph{one-dimensional} if its every maximal block of quantifiers 
leaves at most one variable free. E.g., $\neg \exists 
yz R(x,y,z))$ is one-dimensional, and $\neg \exists z R(x,y,z))$ is not.  By \emph{one-dimensional guarded negation fragment}, \GNFOonedim{} we mean
the subset of \GNFO{} containing its all one-dimensional formulas. Not all \UNFO{} formulas are one-dimensional, but they
can be easily converted to the already mentioned UN-normal form \cite{StC13}, which contains only one-dimensional formulas. The cost of this conversion
is linear. This allows us to view \UNFO{} as
a fragment of \GNFOonedim{}. 

We can define the one-dimensional restriction  \GNFOEQonedim{} of \GNFOEQ{} in a natural way.
We note that moving from \UNFOEQ{} to \GNFOEQonedim{} significantly increases the expressive power. 
An example formula which is in \GNFOEQonedim{} but is not expressible
in \UNFOEQ{} is $\neg \exists xy (R(x,y) \wedge \neg E_1(x,y))$, which says that $R \subseteq E_1$.
Observe, however, that since guards must belong to $\sigma_{\cBase}$ we are not able to express the containment of one equivalence relation in another equivalence, or in a relation from $\sigma_{\cBase}$. 

Our proof from Section 4 can be adapted to cover the case of \GNFOEQonedim{}. The adaptation is not difficult.
What is crucial is that  in the current construction, during the step of providing witnesses, we build isomorphic copies of whole witness structures, which means that we preserve not only positive atoms but also their negations. Thus, we preserve witness structures for \GNFOEQonedim{}.
\begin{theorem} \label{c:onedim}
\GNFOEQonedim{} has a doubly exponential finite model property, and its satisfiability (= finite satisfiability) problem is
\TwoExpTime-complete. 
\end{theorem}

\begin{proof}
Using the standard Scott translation we can transform any \GNFOEQonedim{} sentence into a normal form sentence $\phi$ of the shape
as in (\ref{eq:nf}), where the $\phi_i$ are quantifier-free \GNFO{} formulas.\footnote{We remark here that, since our normal form is
one-dimensional, this is not possible for full \GNFOEQ.} Assume $\str{A} \models \phi$.
First, we need a slightly stronger version of condition (1) in Lemma \ref{l:homomorphisms}---each of the considered homomorphisms should additionally be an isomorphism when restricted to a guarded substructure.
After that we construct a regular tree-like model $\str{A}' \models \phi$, adapting the construction from the proof of Lemma \ref{l:regular} by
extending the notion of declaration so that it treats a subformula of the form $\gamma(\bar{x},\bar{y})\wedge\neg\phi'(\bar{y})$ like an atomic formula.
Finally we apply, without any changes, the construction from the proof of Lemma \ref{l:finiteeq} to $\str{A}'$ and $\phi$ obtaining eventually a finite structure $\str{A}''$. Note that during the step of providing witnesses we build isomorphic copies of whole witness structures, which means we preserve not only positive atoms but also their negations. Thus the elements of $A''$ have all witness structures
required by $\phi$. Consider now the conjunct $\forall x_1, \ldots, x_t \neg \phi_0(\bar{x})$, and take arbitrary elements $a_1, \ldots, a_t \in A''$.
From Lemma \ref{l:finiteeq} we know that there is a homomorphism $\fh:\str{A}'' \restr \{a_1, \ldots, a_t\} \rightarrow \str{A}'$ preserving $1$-types.
If  $\gamma(\bar{z}, \bar{y}) \wedge \neg \phi'(\bar{y})$ is a subformula of $\phi_0$ with $\gamma$ a $\sigma_{\cBase}$-guard and $\str{A}'' \models \gamma (\bar{b}, \bar{c}) \wedge \neg \phi'(\bar{c})$ for some
$\bar{b}, \bar{c} \subseteq \bar{a}$ then, by our construction, all elements of $\bar{b} \cup \bar{c}$ are members of the same witness structure.
As mentioned above such witness structures are isomorphic copies of substructures from $\str{A}$ and $\fh$ works on them as isomorphism, and thus $\fh$  preserves on $\bar{c}$ not only
$1$-types and positive atoms but also negations of atoms in witnesses structures. Since  $\str{A}' \models \neg \phi_0(\fh(a_1), \ldots, \fh(a_t))$ this means that $\str{A}'' \models \neg \phi_0(a_1, \ldots, a_t)$. \qed
\end{proof}

\section{Conclusion} We proved the finite model property for \UNFO{} with equivalence relations and for the one-dimen\-sional
restriction of \GNFO{} with equivalences outside guards. This implies the decidability of the finite satisfiability problem 
for these logics.  
In our forthcoming paper \cite{DK18c} we study the related finite satisfiability problem for \UNFO{} with
arbitrary transitive relations, proving that it is decidable as well.
An interesting direction for further research is  the decidability of finite satisfiability of  full \GNFO{} with equivalences
on non-guard positions. 
	\section*{Acknowledgements}
	This work is supported by {Polish National Science Centre}{} grant No {2016/21/B/ST6/01444}{}.

\bibliographystyle{plain}
\bibliography{mybibshort}

\end{document}

%% file: macrosEQ.tex
\newif\ifnocomment
\newcommand{\cE}{\mathcal{E}}%
\newcommand{\cP}{\mathcal{P}}%

\newcommand{\ff}{\mathfrak{f}}
\newcommand{\fp}{\mathfrak{p}}
\newcommand{\fh}{\mathfrak{h}}

\newcommand{\FOt}{\mbox{$\mbox{\rm FO}^2$}}
\newcommand{\GF}{\mbox{$\mbox{\rm GF}$}}
\newcommand{\GFt}{\mbox{$\mbox{\rm GF}^2$}}
\newcommand{\GFtEG}{\mbox{$\mbox{\rm GF}^2{+}{\rm EG}$}}

\newcommand{\UNFO}{\mbox{UNFO}}
\newcommand{\GNFO}{\mbox{GNFO}}
\newcommand{\GNFOonedim}{\mbox{GNFO}$_1$}

\newcommand{\FF}{\mbox{FF}}
\newcommand{\UNFOt}{\mbox{UNFO$^2$}}
\newcommand{\UNFOEQ}{\mbox{UNFO}{+}\mbox{EQ}}
\newcommand{\UNFOTR}{\mbox{UNFO}{+}\mbox{TR}}
\newcommand{\UNFOtEQ}{\mbox{UNFO$^2$}{+}\mbox{EQ}}

\newcommand{\GNFOTR}{\mbox{BGNFO}{+}\mbox{TR}}

\newcommand{\GNFOEQ}{\mbox{BGNFO}{+}\mbox{EQ}}

\newcommand{\GNFOEQonedim}{\mbox{BGNFO}$_1${+}\mbox{EQ}}

\newcommand{\NExpTime}{\textsc{NExpTime}}
\newcommand{\TwoExpTime}{2\textsc{-ExpTime}}

\newcommand{\str}[1]{{\mathfrak{#1}}}
\newcommand{\restr}{\!\!\restriction\!\!}
\newcommand{\cutout}[1]{}

\renewcommand{\phi}{\varphi} % Nicer-looking phi

\newcommand{\AAA}{\mbox{\large \boldmath $\alpha$}}

\newcommand{\GGG}{\mbox{\large \boldmath $\gamma$}}
\newcommand{\type}[2]{{\rm atp}^{{#1}}({#2})}
\newcommand{\gtype}[2]{{\rm gtp}^{{#1}}({#2})}
\newcommand{\sss}{\scriptscriptstyle}
\newcommand{\cBase}{{\sss\rm base}}%
\newcommand{\cDist}{{\sss\rm dist}}

\definecolor{darkgreen}{rgb}{0.05, 0.50, 0.06}

\definecolor{afb}{rgb}{0.36, 0.54, 0.66}
\definecolor{bro}{rgb}{0.59, 0.29, 0.0}
\definecolor{dmb}{rgb}{0.55, 0.0, 0.0}
\newcommand{\origin}{origin}

\newcommand{\dom}{\mathrm{Dom}}%EK

%% file: DK18aczysty.bbl
\begin{thebibliography}{10}

\bibitem{ABBB16}
A.~Amarilli, M.~Benedikt, P.~Bourhis, and M.~{Vanden Boom}.
\newblock Query answering with transitive and linear-ordered data.
\newblock In {\em International Joint Conference on Artificial Intelligence,
  {IJCAI} 2016}, pages 893--899, 2016.

\bibitem{ABN98}
H.~Andr{\'e}ka, J.~van Benthem, and I.~N\'{e}meti.
\newblock Modal languages and bounded fragments of predicate logic.
\newblock {\em J. Philosophical Logic}, 27:217--274, 1998.

\bibitem{BtCS15}
V.~B{\'{a}}r{\'{a}}ny, B.~ten Cate, and L.~Segoufin.
\newblock Guarded negation.
\newblock {\em J. {ACM}}, 62(3):22, 2015.

\bibitem{BDM11}
M.~Boja{\'{n}}czyk, C.~David, A.~Muscholl, T.~Schwentick, and L.~Segoufin.
\newblock Two-variable logic on data words.
\newblock {\em {ACM} Trans. Comput. Log.}, 12(4):27, 2011.

\bibitem{BMS09}
M.~Bojanczyk, A.~Muscholl, T.~Schwentick, and L.~Segoufin.
\newblock Two-variable logic on data trees and xml reasoning.
\newblock {\em J. ACM}, 56(3), 2009.

\bibitem{DK18}
D.~Danielski and E.~Kiero\'nski.
\newblock Unary negation fragment with equivalence relations has the finite
  model property.
\newblock In {\em Logic in Computer Science, {LICS} 2018}, 2018.

\bibitem{DK18c}
Daniel Danielski and Emanuel Kieronski.
\newblock Finite satisfiability of unary negation fragment with transitivity.
\newblock {\em CoRR}, abs-1802-01318, 2018.

\bibitem{Dzi17}
M.~Dzieciolowski.
\newblock Satisfability issues for unary negation logic.
\newblock Bachelor's thesis, University of Wroc\l{}aw, 2017.

\bibitem{GKV97}
E.~Gr{\"a}del, P.~Kolaitis, and M.~Y. Vardi.
\newblock On the decision problem for two-variable first-order logic.
\newblock {\em B. Symb. Log.}, 3(1):53--69, 1997.

\bibitem{HK14}
L.~Hella and A.~Kuusisto.
\newblock One-dimensional fragment of first-order logic.
\newblock In {\em Advances in Modal Logic, {AIML} 2014)}, pages 274--293, 2014.

\bibitem{JLM18}
J.~Ch. Jung, C.~Lutz, M.~Martel, and T.~Schneider.
\newblock Querying the unary negation fragment with regular path expressions.
\newblock In {\em International Conference on Database Theory, {ICDT} 2018},
  pages 15:1--15:18, 2018.

\bibitem{Kaz06}
Y.~Kazakov.
\newblock {\em Saturation-based decision procedures for extensions of the
  guarded fragment}.
\newblock PhD thesis, Universit\"at des Saarlandes, Saarbr\"ucken, Germany,
  2006.

\bibitem{Kie05}
E.~Kiero\'{n}ski.
\newblock Results on the guarded fragment with equivalence or transitive
  relations.
\newblock In {\em Computer Science Logic}, volume 3634 of {\em LNCS}, pages
  309--324. Springer, 2005.

\bibitem{KMP-HT14}
E.~Kiero\'{n}ski, J.~Michaliszyn, I.~Pratt-Hartmann, and L.~Tendera.
\newblock Two-variable first-order logic with equivalence closure.
\newblock {\em SIAM J. Comput.}, 43(3):1012--1063, 2014.

\bibitem{KO12}
E.~Kiero\'{n}ski and M.~Otto.
\newblock Small substructures and decidability issues for first-order logic
  with two variables.
\newblock {\em J. Symb. Log.}, 77:729--765, 2012.

\bibitem{KT18}
E.~Kiero\'{n}ski and L.~Tendera.
\newblock Finite satisfiability of the two-variable guarded fragment with
  transitive guards and related variants.
\newblock {\em ACM Trans. Comput. Logic}, 19(2):8:1--8:34, 2018.

\bibitem{MPS16}
A.~Montanari, M.~Pazzaglia, and P.~Sala.
\newblock Adding one or more equivalence relations to the interval temporal
  logic.
\newblock {\em Theor. Comput. Sci.}, 629:116--134, 2016.

\bibitem{Mor75}
M.~Mortimer.
\newblock On languages with two variables.
\newblock {\em Zeitschr. f. Math. Logik und Grundlagen d. Math.}, 21:135--140,
  1975.

\bibitem{P-H18}
I.~Pratt-Hartmann.
\newblock The finite satisfiability problem for two-variable, first-order logic
  with one transitive relation is decidable.
\newblock {\em To appear in Math. Log. Q.}, 2018.

\bibitem{P-HST16}
I.~Pratt{-}Hartmann, Wieslaw Szwast, and Lidia Tendera.
\newblock Quine's fluted fragment is non-elementary.
\newblock In {\em Computer Science Logic, {CSL} 2016}, pages 39:1--39:21, 2016.

\bibitem{Sco62}
D.~Scott.
\newblock A decision method for validity of sentences in two variables.
\newblock {\em J. Symb. Log.}, 27:477, 1962.

\bibitem{ST04}
W.~Szwast and L.~Tendera.
\newblock The guarded fragment with transitive guards.
\newblock {\em Ann. Pure Appl. Logic}, 128:227--276, 2004.

\bibitem{StC13}
B.~ten Cate and L.~Segoufin.
\newblock Unary negation.
\newblock {\em Log. Meth. Comput. Sci.}, 9(3), 2013.

\end{thebibliography}
